\documentclass[a4paper,reqno, 11pt]{amsart}
\usepackage{amssymb, graphicx, enumerate, color} 
\usepackage[numbers]{natbib}
\usepackage{amssymb}
\usepackage{amsmath}
\usepackage{dsfont} 
\usepackage{hyperref}
\usepackage{xcolor}
\usepackage{mathtools}
\usepackage{breqn}

\hypersetup{
    colorlinks,
    linkcolor={red!50!black},
    citecolor={blue!50!black},
    urlcolor={blue!80!black}
}

\newtheorem{theorem}{Theorem}[section]

\newtheorem{corollary}[theorem]{Corollary}

\newtheorem{lemma}[theorem]{Lemma}

\newtheorem{definition}{Definition}

\newcommand{\DEF}{\sl}
\newcommand{\StackMST}{{Stackelberg Minimum Spanning Tree}}
\newcommand{\StackMatroid}{{Stackelberg Matroid}}
\newcommand{\greedy}{\mathrm{greedy}}

\newcommand{\R}{\mathbb{R}}
\newcommand{\OPT}{\mathrm{OPT}}
\renewcommand{\leq}{\leqslant}
\renewcommand{\geq}{\geqslant}

\begin{document}

\title[Assortment optimisation under a general discrete choice model]{Assortment optimisation under a general discrete choice model: A tight analysis of revenue-ordered assortments}

\author[G.~Berbeglia]{Gerardo Berbeglia}

\address[G.~Berbeglia]{Melbourne Business School \\
  The University of Melbourne\\
  Melbourne\\
  Australia}

\email{g.berbeglia@mbs.edu}

\author[G.~Joret]{Gwena\"{e}l Joret}

\address[G.~Joret]{Computer Science Department\\
  Universit\'e Libre de Bruxelles\\
  Brussels\\
  Belgium}

\email{gjoret@ulb.ac.be}

\thanks{G.\ Joret was supported by a DECRA Fellowship from the Australian Research Council during part of the project.
Extended abstract based on this work to appear in the proceedings of the eighteenth ACM conference on Economics and Computation (ACM EC’17).}

\date{\today}

\sloppy

\begin{abstract}
  The assortment problem in revenue management is the problem of deciding which subset of products to offer to consumers in order to maximise revenue. A simple and natural strategy is to select the best assortment out of all those that are constructed by fixing a threshold revenue $\pi$ and then choosing all products with revenue at least $\pi$. This is known as the {\em revenue-ordered assortments} strategy.
  In this paper we study the approximation guarantees provided by revenue-ordered assortments when customers are rational in the following sense: the probability of selecting a specific product from the set being offered cannot increase if the set is enlarged. This rationality assumption, known as {\em regularity}, is satisfied by almost all discrete choice models considered in the revenue management and choice theory literature, and in particular by random utility models.
  The bounds we obtain are tight and improve on recent results in that direction, such as for the Mixed Multinomial Logit model by \citet{rusmevichientong2014assortment}.
  An appealing feature of our analysis is its simplicity, as it relies only on the regularity condition.

  We also draw a connection between assortment optimisation and two pricing problems called {\em unit demand envy-free pricing} and {\em Stackelberg minimum spanning tree}: These problems can be restated as assortment problems under discrete choice models satisfying the regularity condition, and moreover revenue-ordered assortments correspond then to the well-studied {\em uniform pricing} heuristic. When specialised to that setting, the general bounds we establish for revenue-ordered assortments match and unify the best known results on uniform pricing.
\end{abstract}

\maketitle

\section{Introduction}
{\DEF Revenue management} consists of a set of methodologies permitting firms to decide on the availability and the price of their products and services. The development of this field began in the late 1970's in the airline industry, and has since been expanding constantly its practices into a large variety of  markets such as grocery stores, retailing, railways, car rentals, accommodation, cruises, and more recently, electronic goods \citep{talluri2006theory}.

At the core of revenue management lies the {\DEF assortment problem}, that of choosing an optimal subset of products/services to offer to consumers in order to maximise the firm's profits.
As an illustration, consider the case of a grocery store that has limited space for its coffee products.
Say it has space for at most 15 different coffee products on the shelves but can choose between 300 products from its distributors (due to the combinations of coffee brands, coffee types, package sizes, etc). Given products' costs and consumer demand, what is the best subset to offer?

In order to solve the assortment problem, it is necessary to know (or at least be able to approximate) the consumer demand for each product, as a function of the assortment of products that are offered. This issue has been widely studied in {\DEF discrete choice theory}, a field which essentially tries to predict the choices of individuals when they select a product from a finite set of mutually exclusive alternatives, typically known as the {\DEF choice set}.  Classical economic theory postulates that individuals select an alternative by assigning a real number known as {\DEF utility} to each option and then choosing the alternative from the choice set that has the maximum utility. Individuals are thus said to be {\DEF utility maximisers}. Different assumptions about the distribution of the product utilities give rise to different discrete choice models. Prominent examples are the {\DEF multinomial logit} model (MNL model) \citep{luce1959} and the more general {\DEF Mixed} MNL model \citep{rusmevichientong2014assortment}.

There is a large literature on studies and methodologies for solving the assortment problem under different discrete choice models such as the independent demand model, the MNL model \citep{talluri2004revenue}, and {\DEF Mixed} MNL model \citep{rusmevichientong2014assortment}. A well-known heuristic is the {\DEF revenue-ordered assortments} strategy. It consists in selecting the best assortment out of all those that are constructed by fixing a threshold $\pi$ and then selecting all products whose revenue is at least $\pi$. This strategy is appealing for two reasons. First, it only needs to evaluate as many assortments as there are different revenues among products, independently of the number of products the firm offers. Second, even if one has no knowledge about consumer choice behaviour, revenue-ordered assortments can still be used, one only needs to be able to evaluate the revenue of these revenue-ordered assortments. This second point is important in practice since most of the time not only the parameters of the assumed discrete choice model are not known but also it is not known what type of discrete choice model the consumers are following (e.g.\ MNL model, Nested MNL model, etc.). This motivates the study of the performance of revenue-ordered assortments under different discrete choice models. \citet{talluri2004revenue} showed for instance that, when consumers follow the MNL model, the revenue-ordered assortments strategy is in fact optimal. In general however, this strategy does not always produce an optimal solution to the assortment problem.

\subsection{Contributions}\label{sec:contributions}
Our first contribution is an analysis of the performance of the revenue-ordered assortments strategy making only minimal assumptions about the underlying discrete choice model: We assume that consumers behave rationally, in the sense that the probability of choosing a specific product $x\in S$ when given a choice set $S$ cannot increase if $S$ is enlarged. This rationality assumption, known as {\em regularity}, is satisfied by almost all models studied in the revenue management and choice theory literature.
This includes in particular all random utility models, as well as other models introduced recently such as the additive perturbed utility model, the hitting fuzzy attention model, and models obtained using a non-additive random utility function (see Section~\ref{section_applications} for a discussion of these models).
We provide three types of revenue guarantees for revenue-ordered assortments: If there are $k$ distinct revenues $r_1, r_2, \dots, r_k$ associated with the products (listed in increasing order), then revenue-ordered assortments approximate the optimum revenue to within a factor of
\begin{enumerate}[(A)]
\item $1/k$;
\item $1 / (1 + \ln(r_k / r_1))$, and
\item $1 / (1 + \ln \nu)$,
\end{enumerate}
where $\nu$ is defined with respect to an optimal assortment $S^*$ as the ratio between the probability of just buying a product and that of buying a product with highest revenue in $S^*$. These three guarantees are in general incomparable, that is, (A), (B), or (C) can be the largest depending on the instance.

When applied to the special case of Mixed MNL models, bound (B) improves the recent analysis of revenue-ordered assortments by \citet{rusmevichientong2014assortment}, who showed a bound of   $1 / (e(1 + \ln(r_k / r_1)))$.
Our results also provide the first revenue guarantees for a tractable approach to solving the assortment problem under some recent discrete choice models that lie beyond the random utility model. In particular, in Section~\ref{section_beyond_rum} we briefly describe the \emph{Additive Perturbed Utility Model} \citep{fudenberg2015stochastic}; the {\em Hitting Fuzzy Attention Model} \citep{aguiar2015stochastic}, and a model proposed by \citet{mcclellon2015non} based on a capacity function. Since these three models satisfy the regularity property, our revenue guarantees for the revenue-ordered assortment strategy hold.

 After finishing a preliminary version of this paper, we learned of independent results by \citet{aouad2015approximability} on the approximation guarantees of revenue-ordered assortments. In the important case of random utility models, they showed a bound of $\Omega (1 /  \ln(r_k / r_1))$, which is thus within a constant factor of our bound (B). (In fact, one can deduce a bound of $2 /  (5 \lceil \ln(r_k / r_1)\rceil)$ from their proof.) They also proved a bound of $\Omega (1 / \ln \tilde \lambda)$, where $\tilde \lambda$ is the probability of buying a product when offered the set consisting of only the products with highest revenue. This is closely related to bound (C) above, as it can be shown that $\nu \leq \tilde \lambda$.

Complementing our analysis, we show that the three bounds (A), (B), and (C) are exactly tight, in the sense that none of the bounds remains true if multiplied by a factor $(1+ \epsilon)$ for any $\epsilon > 0$.
Let us remark that bounds (A) and (B) also provide a quick and easy way to obtain some upper bound on the optimum revenue that can be achieved on a given instance, by simply checking the revenue provided by revenue-ordered assortments. While the resulting upper bounds can of course be far from the optimum, they have the merit of being straightforward to compute, independently of how complex the underlying discrete choice model is, as long as it satisfies the regularity condition. Bound (C) on the other hand is typically hard to compute exactly as it involves knowing an optimal solution; nevertheless, some non-trivial bounds that are easy to compute can be deduced from (C) by bounding $\nu$ from above (using the bound $\nu \leq \tilde \lambda$ mentioned above for instance).

Our second contribution is to draw a connection between assortment optimisation and some pricing problems studied in the theoretical computer science literature by showing that these pricing problems can be restated as an assortment problem under a discrete choice model satisfying the above-mentioned rationality assumption. This includes unit demand envy-free pricing problems and the Stackelberg minimum spanning tree problem. Unlike the assortment problem, these problems consist not in selecting a subset of products (that have an attached price) to offer to consumers, but rather to assign a price to each product (see Section~\ref{section_applications} for definitions). The best known approximation algorithm for these problems consists in assigning the same (well-chosen) price to all the products, and is called uniform pricing. We will show that these pricing problems can be seen as special cases of the assortment problem under a regular discrete choice model, and prove that the revenue-ordered assortment strategy applied to these special cases is equivalent to applying the uniform pricing algorithm to these pricing problems.

We conclude the paper with a brief analysis of the single-leg multi-period setting with limited capacity as studied by \citet{talluri2004revenue}. In this problem the seller has $Q$ units to sell over a finite time horizon $T$. ~\citet{rusmevichientong2014assortment} showed that when the seller restricts to offer revenue-ordered assortments and consumers follow a Mixed MNL choice model: (i) the offer set gets larger (or stays the same) as the number of time periods remaining increases, and (ii) the offer set gets smaller (or stays the same) as the number of available units remaining decreases. We observe that these two monotonicity results hold more generally for the discrete choice models considered in this paper. As mentioned by~\citet{rusmevichientong2014assortment}, these results could potentially be used in the implementation of standard revenue management systems.

\subsection{Related literature}
The assortment problem is an active research topic in revenue management, and providing an exhaustive literature review is beyond the scope of this paper.
In this section, we focus mostly on previous works that are directly related to our contributions.

One of the first studies of the assortment problem for a general discrete choice model was carried out by \citet{talluri2004revenue}. In that paper, the authors considered a single-leg seat allocation problem in which a firm sells aircraft seats to consumers arriving one at a time. Each consumer selects at most one fare among the ones that are offered, and the firm has to decide the subset of fares to offer at each time period, depending on the number of available seats and time periods remaining. The authors have shown that this problem reduces to solving a static (or single shot) assortment problem in which one wishes to maximise the expected profit on a single consumer without caring about capacity. For the special case in which the consumer choice model is the Multinomial Logit (MNL) model, they proved that the optimal choice sets are revenue-ordered assortments. Thus, solving the assortment problem when consumers follow a MNL model can be done efficiently in polynomial time. Again under the MNL model, \citet{rusmevichientong2010dynamic} studied the assortment problem subject to the constraint that there is maximum number of products that can be shown in the assortment. Although the assortment sometimes fails to be a revenue-ordered assortment, the authors proved than an optimal assortment can still be found in polynomial-time. Another extension of the assortment problem over MNL model was considered by \citet{rusmevichientong2012robust} where the authors formulated the problem as a robust optimization problem in which the true parameters are unknown. Recently, \citet{feldman2018taking} ran a large field experiment on the online platform Alibaba and reported that using an MNL choice model to optimise assortments provides higher revenues than a machine learning method that is currently being used on Alibaba.

Another series of papers studied the assortment problem in which different shares of customers follow different MNL models, a model known as the Mixed MNL model. The assortment problem under a Mixed MNL model is NP-hard  \citep{bront2009column}. A branch-and-cut algorithm was presented by \citet{mendez2014branch}, and computational
methods to obtain good upper bounds on the optimal revenue were given in \citet{feldman2015bounding}.
\citet{rusmevichientong2014assortment} proved that the problem remains NP-hard even when the model is composed of only two customer types. The authors also studied the performance of revenue-ordered assortments in this setting, and proved in particular an approximation ratio of  $1 / (e(1 + \ln(r_k / r_1)))$, as mentioned earlier.

\citet{blanchet2013markov} introduced a new discrete choice model where consumers preferences are built using a Markov chain in which states correspond to products. They showed that the assortment problem under this discrete choice model can be solved in polynomial time. \citet{feldman2014revenue} extended their results to the case of a single-leg seat allocation problem over a finite time horizon when consumers follow the Markov chain model. \citet{desir2015capacity} proved that the assortment problem with a capacity constraint in the Markov chain model is APX-hard and provide a polynomial-time constant-factor approximation algorithm. Recently, \citet{berbeglia2016discrete} showed that every choice model based on Markov chains models belongs to the class of choice models based on random utility. \citet{jagabathula2014assortment} introduced a local search heuristic for the assortment problem under an arbitrary discrete choice model. The author proved that the heuristic is optimal in the case of the MNL model, and remains so even if the choice set is subject to a maximum cardinality constraint. \citet{aouad2015assortment} study the assortment problem under a family of choices models known as {\DEF consider-then-chose} models that have been empirically tested in the marketing literature. The authors provided several computational complexity results as well as a dynamic programming algorithm that can be used by heuristics for arbitrary choice models based on random utility. Recently, there has been progress in studying choice models that incorporate position biases, i.e. models where consumers choices are affected by the specific positions or configurations in which the products are offered \citep{abeliuk2015assortment, aouad2015display, davis2013assortment}.

The strongest negative result to date regarding the computational complexity of the assortment problem is due to \citet{aouad2015approximability}. They  proved that the assortment problem under a random utility model is NP-hard to approximate to within a factor of $\Omega(1/n^{1-\epsilon})$ where $n$ denotes the number of products, and to within a factor of $\Omega(1 / \log^{1-\epsilon}(r_k / r_1))$, for every $\epsilon > 0$. Note that the first result implies in particular the hardness of achieving an approximation ratio of $\Omega(1 / k^{1-\epsilon})$, since  $n \geq k$. Comparing this with approximation guarantees (A) and (B) for revenue-ordered assortments, we thus see that the latter heuristic achieves essentially the best possible approximation ratios (w.r.t.\ these parameters) among all computationally efficient strategies. We note that revenue guarantees that are functions of other parameters of the model were given by \citet{aouad2015approximability} in the case of random utility models.

Another line of research related to our work is the study of {\DEF envy-free pricing problems}. Envy-free pricing problems were introduced by \citet{rusmevichientong2003non} (see section \ref{section_envy_free_pricing} for the definition). \citet{aggarwal2004algorithms} and \citet{guruswami2005profit} analysed some natural pricing algorithms, while \citet{briest2011buying} and \citet{chalermsook2012improved} proved inapproximability results which show essentially that these simple algorithms are probably the best one can hope for among those that are computationally efficient (i.e.\ that run in polynomial time). Envy-free pricing will be the focus of Section~\ref{section_envy_free_pricing}. As mentioned in Section \ref{sec:contributions}, we will show that some of these problems can be seen as special cases of the assortment problem under a regular discrete choice model, thus connecting the envy-free pricing literature with the operations research area of revenue management.

\section{The assortment problem} \label{section_general_assortment_problem}
Think of each alternative (or choice) in $\mathcal{C}=\{1,\dots,N\}$ as products types of a firm that are available to sell to consumers. Faced with a {\DEF choice set} $S \subseteq \mathcal{C}$, consumers choose their most preferred product out of the set $S$, or simply choose not to purchase at all. Since consumers may have heterogeneous preferences, we let $\mathcal{P}(x,S)$ denote the probability that a consumer will choose product $x$ when faced with a choice set $S \subseteq \mathcal{C}$. (Defining choices as probabilities is also required to account for choice functions with a stochastic component, even in the case consumers have homogeneous (stochastic) preferences.) Following~\citet{talluri2004revenue} we let $x=0$ denote the no-purchase option. Therefore, $\mathcal{P}(0, S) = 1 - \sum_{x \in S}\mathcal{P}(x, S)$.
A {\DEF regular discrete choice model} is characterised by the function $\mathcal{P}$ (called {\DEF the system of choice probabilities}) defined over the domain $I=\{(x,S): S \subseteq \mathcal{C}, x \in S \cup \{0\} \}$, and satisfies the following four axioms:

\newcounter{saveenum}
\begin{enumerate}[(i)]
\itemsep.6ex
\item $\mathcal{P}(x, S) \geq 0 $ for every $x\in \mathcal{C} \cup \{0\} $ and $S\subseteq \mathcal{C}; \label{discrete_ineq_1}$
\item $\mathcal{P}(x, S) = 0 $ for every $ x\in \mathcal{C} $ and $ S\subseteq \mathcal{C} \setminus \{x\}; \label{discrete_eq_1}$
\item \label{axiom:probability} $\sum_{x \in S}\mathcal{P}(x, S)  \leq 1 $ for every $S\subseteq \mathcal{C}. \label{discrete_eq_2}$
\item $\mathcal{P}(x, S) \geq \mathcal{P}(x, S') $ for every $ S\subseteq S'\subseteq \mathcal{C} $ and $ x\in S \cup \{0\}$. \label{axiom:regularity}
\setcounter{saveenum}{\value{enumi}}
\end{enumerate}

\setcounter{equation}{4}

Inequality (\ref{discrete_ineq_1}) states that probabilities must be non negative. Equality (\ref{discrete_eq_1}) captures the fact that consumers cannot choose a product which is not offered. Inequality (\ref{discrete_eq_2}) specifies that consumers can choose at most one product from the choice set. Finally, inequality (\ref{axiom:regularity}) ensures that the probability of choosing a specific product does not increase when the choice set is enlarged.
This last axiom is called the {\DEF regularity} axiom.

Note that axioms  (\ref{discrete_ineq_1})--(\ref{discrete_eq_2}) are satisfied by any discrete choice model.
As we will see, the regularity axiom (axiom (\ref{axiom:regularity})) is satisfied by almost all models studied in revenue management, economics, and psychology, and in particular by the well-known {\DEF random utility models} (RUM); see Section \ref{section_applications}. We remark that there nevertheless exist discrete choice models for which the regularity axiom fails to hold, this is the case for instance for a perception-based model \citep{echenique2013} and the pairwise choice markov chain model \citep{ragain2016pairwise}.

We now proceed with the definition of the assortment problem. Consider a regular discrete choice model with system of choice probabilities $\mathcal{P}$, and let $r: \mathcal{C} \rightarrow \mathbb{R}_{>0}$ be a revenue function associating to each element of $\mathcal{C}$ a positive (per unit) revenue or price. (In this paper costs are assumed to be negligible and therefore the words `revenue' and (selling) `price' are used interchangeably.)
\begin{definition}
For each set $S \subseteq \mathcal{C}$, the seller's revenue when offering set $S$ is  $$\sum_{x\in S} \mathcal{P}(x, S) r(x).$$
The {\DEF assortment problem} consists in finding a subset $S$ of the products in $\mathcal{C}$ so that the corresponding revenue is maximised.
We let $\OPT$ denote the maximum revenue that can be achieved.
\end{definition}

One may thus see the assortment problem as that of finding the best choice set $S$ maximising the expected utility.
It can be shown that, under reasonable assumptions on how the system of choice probabilities $\mathcal{P}$ is provided in input, the assortment problem is NP-hard. In fact, it remains NP-hard even in some very restricted cases such as when the discrete choice model is a mixture of only two multinomial logit models \citep{rusmevichientong2014assortment}.

We end this section with a straightforward but important observation about regular discrete choice models, namely that the probability of making a purchase does not decrease when the choice set is enlarged.
\begin{lemma} \label{lemma_property_regularDCM}
If $\mathcal{P}$ denotes the system of choice probabilities of a regular discrete choice model, then
\begin{eqnarray*}
\sum_{x \in S}\mathcal{P}(x, S) \leq \sum_{x\in S'}\mathcal{P}(x, S')
\end{eqnarray*}
for every $S\subseteq S'\subseteq \mathcal{C}$.
\end{lemma}
\begin{proof}
Let $S\subseteq S'\subseteq \mathcal{C}$.
We have
$$
\sum_{x \in S}\mathcal{P}(x, S)  + \mathcal{P}(0, S) =  \sum_{x \in S'}\mathcal{P}(x, S')  + \mathcal{P}(0, S') =1
$$
by definition of $\mathcal{P}(0,S)$ and axiom~(\ref{discrete_eq_2}), and thus
$$
\sum_{x \in S}\mathcal{P}(x, S)  - \sum_{x \in S'}\mathcal{P}(x, S') = \mathcal{P}(0, S') - \mathcal{P}(0, S) \leq 0,
$$
by the regularity axiom~(\ref{axiom:regularity}), implying
$\sum_{x \in S}\mathcal{P}(x, S) \leq \sum_{x \in S'}\mathcal{P}(x, S')$.
\end{proof}

\section{Performance guarantees of revenue-ordered assortments}
\label{section_general_performance_guarantees}
Fix a regular discrete choice model with
system of choice probabilities $\mathcal{P}$,
a revenue function $r: \mathcal{C} \rightarrow \mathbb{R}_{> 0}$, and consider the corresponding assortment problem.
Let us recall the {\DEF revenue-ordered assortments} strategy to obtain a (hopefully good) solution:
Let $r_{1}, r_{2}, \dots, r_{k}$ be the distinct values taken by the revenue function $r$, sorted in increasing order; thus, $0 < r_{1} < r_{2} < \cdots < r_{k}$.
For each $i\in [k]$ let $S_{i} \subseteq \mathcal{C}$ be the set consisting of all products of revenue {\em at least} $r_{i}$.
Then simply compare the revenue of each of the $k$ sets $S_{1}, \dots, S_{k}$, and choose one with maximum revenue.

In this section we present an analysis of the approximation guarantees of revenue-ordered assortments.
We give three lower bounds on the approximation ratio which, in general, are incomparable. (For $0 < \alpha \leq 1$, an algorithm is said to achieve an {\DEF approximation ratio} of $\alpha$ for the assortment problem, or equivalently to {\DEF approximate the problem to within a factor of $\alpha$}, if the algorithm always produces a solution whose revenue is at least $\alpha \cdot \OPT$.)
We begin with a simple one:

\begin{theorem} \label{theorem_revenue_k}
Revenue-ordered assortments approximate the optimum revenue to within a factor of $\frac{1}{k}$.
\end{theorem}
\begin{proof}
Let $S^{*} \subseteq \mathcal{C}$ denote an optimal solution to the assortment problem.
Let $j \in [k]$ be the index of a set maximising its revenue among $S_{1}, \dots, S_{k}$;
thus, the revenue provided by the revenue-ordered assortments strategy is
$\sum_{x\in S_{j}} \mathcal{P}(x, S_{j}) r(x)$.

We begin with a technical observation about
the revenue of $S_{i}$ ($i \in [k]$):
\begin{equation}
\label{eq:technical_bound}
\sum_{x\in S_{i}} \mathcal{P}(x, S_{i}) r(x)
\geq \sum_{x\in S^{*} \cap S_{i}} \mathcal{P}(x, S^{*}) r_{i}.
\end{equation}
This can be proved as follows:
\begin{dmath*}[compact]
\sum_{x\in S_{i}} \mathcal{P}(x, S_{i}) r(x)
\geq \sum_{x\in S_{i}} \mathcal{P}(x, S_{i}) r_{i}
\geq \sum_{x\in S^{*} \cap S_{i}} \mathcal{P}(x, S^{*} \cap S_{i}) r_{i}
\geq \sum_{x\in S^{*} \cap S_{i}} \mathcal{P}(x, S^{*}) r_{i}.
\end{dmath*}
Above, the second inequality follows from Lemma~\ref{lemma_property_regularDCM}
and the third follows from the regularity axiom (c.f.\
axiom~\eqref{axiom:regularity}).

Using~\eqref{eq:technical_bound}, it is straightforward to obtain
the lower bound of $1/k$ on the approximation ratio:
\begin{dmath*}[compact]
\OPT = \sum_{x\in S^{*}} \mathcal{P}(x, S^{*}) r(x)
= \sum_{i=1}^{k}\sum_{\substack{x\in S^{*}, \\ r(x) = r_{i}}} \mathcal{P}(x, S^{*}) r_{i}
\leq \sum_{i=1}^{k} \sum_{x\in S^{*} \cap S_{i}} \mathcal{P}(x, S^{*}) r_{i}
\leq \sum_{i=1}^{k} \sum_{x\in S_{i}} \mathcal{P}(x, S_{i}) r(x)
\leq \sum_{i=1}^{k} \sum_{x\in S_{j}} \mathcal{P}(x, S_{j}) r(x)
= k \sum_{x\in S_{j}} \mathcal{P}(x, S_{j}) r(x),
\end{dmath*}
that is,
$$\sum_{x\in S_{j}} \mathcal{P}(x, S_{j}) r(x) \geq \frac{1}{k} \OPT,$$
as desired.
(We note that the second inequality holds by~\eqref{eq:technical_bound} and the third by the definition of the index $j$.)
\end{proof}

We continue with the second bound on the approximation ratio, which is a function of the ratio between the highest and lowest revenues of products in $\mathcal{C}$.

\begin{theorem} \label{theorem_revenue_ratio}
Revenue-ordered assortments approximate the optimum revenue to within a factor of
$$\frac{1}{\sum_{i=1}^{k} \frac{r_{i} - r_{i-1}}{r_{i}}}
\geq \frac{1}{1 + \ln \rho}$$
where $\rho \coloneqq r_{k} / r_{1}$ and $r_{0} \coloneqq 0$.
\end{theorem}
\begin{proof}
As in the previous proof, let $S^{*} \subseteq \mathcal{C}$ denote an optimal solution to the assortment problem, and
let $j \in [k]$ be the index of a set maximising its revenue among $S_{1}, \dots, S_{k}$.
First, we rewrite the revenue of $S^{*}$ as follows.
\begin{dmath*}[compact]
\OPT = \sum_{x\in S^{*}} \mathcal{P}(x, S^{*}) r(x)
= \sum_{i=1}^{k}\sum_{\substack{x\in S^{*}, \\ r(x) = r_{i}}} \mathcal{P}(x, S^{*}) r_{i}. 
\end{dmath*}
Rearranging the terms in the last expression, we obtain:
\begin{dmath*}[compact]
\OPT
= \sum_{\ell=1}^{k}(r_{\ell} - r_{\ell-1})
\sum_{i=\ell}^{k}\sum_{\substack{x\in S^{*}, \\ r(x) = r_{i}}} \mathcal{P}(x, S^{*}) = \sum_{\ell=1}^{k}(r_{\ell} - r_{\ell-1})
\sum_{x\in S^{*} \cap S_{\ell}} \mathcal{P}(x, S^{*})
= \sum_{\ell=1}^{k} \frac{r_{\ell} - r_{\ell-1}}{r_{\ell}}
\sum_{x\in S^{*} \cap S_{\ell}} \mathcal{P}(x, S^{*}) r_{\ell}
\leq \sum_{\ell=1}^{k} \frac{r_{\ell} - r_{\ell-1}}{r_{\ell}}
\sum_{x\in S_{\ell}} \mathcal{P}(x, S_{\ell}) r(x) \\
\leq (1 + \ln \rho)
\sum_{x\in S_{j}} \mathcal{P}(x, S_{j}) r(x),
\end{dmath*}
as desired. Here, the first inequality
holds by~\eqref{eq:technical_bound} (c.f.\ proof of Theorem~\ref{theorem_revenue_k}), and the second one
follows from the observation that
\[
\sum_{\ell=1}^{k} \frac{r_{\ell} - r_{\ell-1}}{r_{\ell}}
= 1 + \sum_{\ell=2}^{k} \frac{r_{\ell} - r_{\ell-1}}{r_{\ell}}
\leq 1 + \int_{r_1}^{r_k} \frac{dt}{t}
= 1 + \ln (r_{k}/r_{1}).
\]
\end{proof}

Next we bound the approximation ratio using a technical quantity which is a function of an optimal solution.
We will see concrete applications of this bound later on.

\begin{theorem}\label{theorem_demand_ratio_in_optimal}
Let $S^{*} \subseteq \mathcal{C}$ denote an optimal solution, and
let $$N_i\coloneqq \sum_{\substack{x\in S^{*}, \\ r(x) \geq r_{i}}} \mathcal{P}(x, S^{*})$$
for each $i \in [k]$.
Suppose that $N_1 >0$, and
let $\ell \in [k]$ be maximum such that $N_{\ell} > 0$.
Then revenue-ordered assortments approximate the optimum revenue to within a factor of
$$\frac{1}{\sum_{i=1}^{\ell} \frac{N_{i} - N_{i+1}}{N_i}} \geq \frac{1}{(1 + \ln \nu)},$$ where $\nu \coloneqq N_{1} / N_{\ell}$.
\end{theorem}
\begin{proof}
Let $j \in [k]$ be the index of a set $S_{i}$
maximising its revenue among $S_{1}, \dots, S_{k}$.
Observe that
$N_i \leq \sum_{x\in S_{i}} \mathcal{P}(x, S_{i}) $
for every $i\in [k]$, and thus
$N_i r_i \leq \sum_{x\in S_{i}} \mathcal{P}(x, S_{i}) r_i $ for every $i$.
It follows
\begin{equation}
\label{eq:Ni}
N_i r_i \leq \sum_{x\in S_{j}} \mathcal{P}(x, S_{j}) r(x)
\end{equation}
for every $i\in [k]$, that is,
the revenue resulting from the revenue-ordered assortments strategy is at least $\max_{1 \leq i \leq k} N_i r_i$.

The revenue of $S^{*}$ can be expressed as follows:
\begin{dmath*}[compact]
\sum_{x\in S^{*}} \mathcal{P}(x, S^{*}) r(x)
= \sum_{i=1}^{k}\sum_{\substack{x\in S^{*}, \\ r(x) = r_{i}}} \mathcal{P}(x, S^{*}) r_{i}
= \sum_{i=1}^{\ell}\sum_{\substack{x\in S^{*}, \\ r(x) = r_{i}}} \mathcal{P}(x, S^{*}) r_{i}
= \sum_{i=1}^{\ell} (N_{i} - N_{i+1}) r_{i}
= \sum_{i=1}^{\ell} \frac{N_{i} - N_{i+1}}{N_i} N_i r_{i}
\end{dmath*}
where we let $N_{k+1} \coloneqq 0$ in case $\ell =k$.
Using~\eqref{eq:Ni}, we then obtain:
\begin{dmath*}[compact]
\sum_{x\in S^{*}} \mathcal{P}(x, S^{*}) r(x)
= \sum_{i=1}^{\ell} \frac{N_{i} - N_{i+1}}{N_i} N_i r_{i}
\leq \sum_{i=1}^{\ell} \frac{N_{i} - N_{i+1}}{N_i} \sum_{x\in S_{j}} \mathcal{P}(x, S_{j}) r(x) \\
\leq (1 + \ln \nu) \sum_{x\in S_{j}} \mathcal{P}(x, S_{j}) r(x),
\end{dmath*}
as desired.
\end{proof}

\begin{theorem} \label{tightness_theorem}
All three bounds on the approximation ratio given in this section are tight.
\end{theorem}
\begin{proof}

Let $k \geq 1$ and let $\mathcal{C}$ consists of $N=\frac{k(k+1)}{2}$ products.
To simplify the description, it will be convenient to identify the set $\mathcal{C}$ of products with the set of all pairs $(i, j)$ with $i\in [k]$ and $j \in [i]$.
Next, fix some $\varepsilon$ with $0 < \varepsilon \leq \frac12$ and let the revenue of product $(i, j)$ be $\varepsilon^{-j}$.
Thus, defining $r_{1}, \dots, r_{k}$ as before, we have $r_{i}=\varepsilon^{-i}$ for each $i\in [k]$.
Finally, define the system of choice probabilities $\mathcal{P}$ by letting
$$\mathcal{P}((i, j), S) \coloneqq \left\{
\begin{array}{ll}
\varepsilon^{i} & \textrm{ if } (i, 1), \dots, (i, j-1) \notin S  \\
0 & \textrm{ otherwise}.
\end{array}\right.$$
for each $S\subseteq \mathcal{C}$ and $(i, j) \in S$.
Also, let $\mathcal{P}(0, S) \coloneqq 1 - \sum_{x \in S}\mathcal{P}(x, S)$ for each $S\subseteq \mathcal{C}$, as expected.

We proceed to verify that $\mathcal{P}$ satisfies all four axioms. Inequality \eqref{discrete_ineq_1} and equality \eqref{discrete_eq_1} are clearly satisfied. For each subset $S \subseteq \mathcal{C}$ it holds that
\begin{align*}
\sum_{(i,j) \in S} \mathcal{P}((i,j),S) \leq \varepsilon^1 + \varepsilon^2 + \hdots + \varepsilon^k  < \frac{1}{1-\varepsilon} - 1 = \frac{\varepsilon}{1-\varepsilon} \leq 1,
\end{align*}
so inequality \eqref{discrete_eq_2} is also respected.
To prove that the regularity axiom holds, consider a choice set $S \subseteq \mathcal{C}$, an element $(i,j) \in S$, and a choice set $S'$ such that $S \subseteq S' \subseteq \mathcal{C}$. Clearly, $\mathcal{P}((i,j),S)$ and $\mathcal{P}((i,j),S')$ can only take the values $\varepsilon^i$ or 0.
Moreover, it can be checked from the definition of $\mathcal{P}$ that $\mathcal{P}((i,j),S)=0$ implies
$\mathcal{P}((i,j),S')=0$. Hence,
\begin{align}\label{regularity_tight_thm_1}
\mathcal{P}((i,j),S) \geq \mathcal{P}((i,j),S') \textrm{ for every } S\subseteq S'\subseteq \mathcal{C}  \textrm{ and } (i,j) \in S.
\end{align}

For a subset $ S \subseteq \mathcal{C}$ and index $i\in [k]$, let $S_i = \{(i,j):  j\in[i], (i,j) \in S \}$, thus $S= \cup_{i=1}^k S_i$ and $S_i \cap S_{i'} = \emptyset$ for all $i,i' \in [k]$ with $i \neq i'$. For
sets $S, S'$ with $S\subseteq S'\subseteq \mathcal{C}$, we have

\begin{align}\label{regularity_tight_thm_2}
\sum_{(i,j) \in S}\mathcal{P}((i,j),S) = \sum_{i\in[k]} \sum_{(i,j) \in S_i}\mathcal{P}((i,j),S_i) = \sum_{i\in[k],S_i \neq \emptyset}\varepsilon^i \leq \sum_{i\in[k],S'_i \neq \emptyset} \varepsilon^i  = \sum_{(i,j) \in S'}\mathcal{P}((i,j),S').
\end{align}

From \eqref{regularity_tight_thm_2} it follows that
\begin{align}\label{regularity_tight_thm_3}
\mathcal{P}(0,S) \geq \mathcal{P}(0,S') \textrm{ for every } S\subseteq S'\subseteq \mathcal{C}.
\end{align}

We deduce from \eqref{regularity_tight_thm_1} and \eqref{regularity_tight_thm_3} that the regularity axiom \eqref{axiom:regularity} holds, and therefore the constructed system of choice probabilities is a regular discrete choice model.
(Alternatively, one could verify that the proposed choice model belongs to the class of choice models based on stochastic preferences and therefore the regularity property is inherited, see Section \ref{sec:rum}.)

Now, the best assortment among all $k$ revenue-ordered assortments is the first one, namely
the one consisting of all products of revenue at least $r_{1}$, that is, all products. Its revenue is
$(\varepsilon + \varepsilon^{2} + \cdots + \varepsilon^{k})\cdot \varepsilon^{-1}
< \frac{1}{1 - \varepsilon}$.
On the other hand, observe that offering the products $(i,i)$ for each $i \in [k]$ yields a revenue of $k$ (in fact, this is the optimal solution).

Thus, if we let $\varepsilon$ tend to $0$ then the ratio between the two values tends to $k$,
showing that Theorem~\ref{theorem_revenue_k} is best possible.
Also, observe that $\sum_{i=1}^{k} \frac{r_{i} - r_{i-1}}{r_{i}}$ tends to $k$
when $\varepsilon \to 0$, showing that Theorem~\ref{theorem_revenue_ratio} is tight as well.
Finally, regarding Theorem~\ref{theorem_demand_ratio_in_optimal}, if we define $N_{i}$ ($i\in [k]$)
w.r.t.\ the optimal solution $S^{*}\coloneqq \{(i,i): i\in [k]\}$, we have
$N_{i} = \varepsilon^{i} + \cdots + \varepsilon^{k}$ for each $i\in [k]$, and
$\sum_{i=1}^{k} \frac{N_{i} - N_{i+1}}{N_i}$ also tends to $k$ when $\varepsilon \to 0$.
\end{proof}

\section{Applications} \label{section_applications}
In this section we describe how the main results from Section~\ref{section_general_performance_guarantees} can be applied to derive revenue guarantees of the revenue-ordered assortments for the assortment problem under a wide range of choice models studied in revenue management as well as in the theoretical economics literature. We highlight that often, in practice, one does not know what choice model consumers are following. Nevertheless, as long as the firm is able to evaluate the expected revenue obtained by each of $k$ the choice sets $S_i$ ($i=1,\hdots,k$) stemming from the revenue-ordered assortment strategy, and the (unknown) choice model satisfies regularity, the revenue guarantees obtained in Section~\ref{section_general_performance_guarantees} hold. We also show that the bounds obtained in Section~\ref{section_general_performance_guarantees} generalise previous results regarding the uniform pricing heuristic for some pricing problems studied in theoretical computer science.

\subsection{Random utility models}\label{sec:rum}

Let us first recall the definition of random utility models first proposed by \citet{thurstone1927law}. In a random utility model, individuals assign a random variable $U_x$ (utility) to each product $x \in \mathcal{C} \cup \{0\}$. These $N+1$ random variables are jointly distributed over $\mathbb{R}^{N+1}$, with a probability measure $\Pr$ such that $\Pr(U_x = U_y)=0$ for all distinct $x,y \in \mathcal{C} \cup \{0\}$. Given a choice set $S \subseteq \mathcal{C}$, an individual considers a realisation $(u_0,u_1,\dots,u_N)$ of the random utilities and then selects
$x \in S \cup \{0\}$ such that $u_{x}$ is maximum (note that this could be the
no-purchase option $x=0$).

\begin{definition} \label{random_utility_characterization}
Suppose $U_0, U_1, \dots, U_N$ and $\Pr$ are  as above.
Then the system of choice probabilities $\mathcal{P}$ induced by this random utility model is obtained by setting
\begin{eqnarray*}
\mathcal{P}(x,S) = \Pr ( U_x = \max \{ U_y : y \in S \cup\{0\} \} )
\end{eqnarray*}
for all $S \subseteq \mathcal{C}$ and  $x \in S \cup\{0\}$.
\end{definition}

Whenever $\mathcal{P}$ is as in Definition~\ref{random_utility_characterization}
we say that $\mathcal{P}$ is a {\DEF random utility based discrete choice model}, which we abbreviate as RUM for short.
We remark that if $\mathcal{P}$ is a RUM then it is a regular discrete choice model; in particular,
$\mathcal{P}$ satisfies the regularity axiom (axiom (\ref{axiom:regularity})).
Informally, this is because if we consider a choice set $S$ and a product $x\in S \cup \{0\}$ then the likelihood that $x$ maximises utility among products in $S \cup \{0\}$ can only decrease if we enlarge the set $S$.
This was originally observed by \citet{luce_suppes_1965}.
We thus have the following lemma.

\begin{lemma}\label{lemma_all_RUMS_are_regular}
Every RUM is a regular discrete choice model.
\end{lemma}

The converse of Lemma~\ref{lemma_all_RUMS_are_regular} is not true. 
One way to observe this is that every RUM has the following submodularity property (a proof is given in the appendix).  

\begin{lemma}\label{lem:submodularity}
In every random utility based discrete choice model, the demand function $f(S)=\sum_{x \in S} \mathcal{P}(x,S)$ is submodular, i.e.\ $f(S'\cup \{x\}) - f(S') \leq  f(S\cup \{x\}) - f(S)$ for every $S \subseteq S' \subseteq \mathcal{C}$ and $x \in \mathcal{C}$. 
\end{lemma}

On the other hand, submodularity of the demand function does not follow from axioms~\eqref{discrete_ineq_1}-\eqref{axiom:regularity}.
This can be seen on the following simple example, adapted from \citet{mcfadden1990stochastic}
(observe that $f(\{1,2,3\})-f(\{1,2\}) = .15 > f(\{1,3\}) -f(\{1\}) = .1$).

$$
\begin{array}{l|llll}
S & \mathcal{P}(0, S) & \mathcal{P}(1, S) & \mathcal{P}(2, S) & \mathcal{P}(3, S) \\
\hline
\{1\} & .5 & .5 & - & - \\
\{2\} & .5 & - & .5 & - \\
\{3\} & .5 & - & - & .5 \\
\{1,2\} & .4 & .3 & .3 & - \\
\{1,3\} & .4 & .3 & - & .3 \\
\{2,3\} & .4 & - & .3 & .3 \\
\{1,2,3\} & .25 & .25 & .25 & .25 \\
\end{array}
$$

A consequence of Lemma~\ref{lemma_all_RUMS_are_regular} is that
the revenue guarantees of revenue-ordered assortments given by Theorems \ref{theorem_revenue_k}, \ref{theorem_revenue_ratio}, and \ref{theorem_demand_ratio_in_optimal} hold for the assortment problem under a RUM, regardless on how complex the RUM in question is.
For example, our results apply to some recent models in neuroscience and psychology that predict how the brain reaches a decision, see \citet{webb2016dynamics} and \citet{webb2013neural}.

It is well known that a random utility based discrete choice model can also be described by means of a probability distribution over all rankings of the elements in $\mathcal{C} \cup \{0\}$ as follows:
Let $\Pr$ be a probability measure over the set $\mathcal{S}_{N+1}$ of the $(N+1)!$ permutations of the elements in $\mathcal{C} \cup \{0\}$, often called a {\DEF stochastic preference}. Given a permutation $\prec \in \mathcal{S}_{N+1}$ and two distinct elements $x, y \in \mathcal{C} \cup \{0\}$, let us write $x \prec y $ whenever $x$ appears before $y$ in $\prec$.

\begin{definition} \label{stochastic_preference_characterization}
Suppose $\Pr$ is  as above.
Then the system of choice probabilities $\mathcal{P}$ induced by this stochastic preference model is obtained by setting
\begin{eqnarray} \label{stochastic_preference_probabilities}
\mathcal{P}(x,S) = \Pr( x \prec y \text{ for all } y \in (S \cup \{0\}) - \{x\})
\end{eqnarray}
for all $S \subseteq \mathcal{C}$ and $x \in S\cup \{0\}$.
\end{definition}

In other words, $\mathcal{P}(x,S)$ is the probability that $x$ is ranked first among all products in $S$ and the no-purchase option $0$.
As mentioned earlier, random utility models and stochastic preference models are essentially equivalent, in the sense that they give rise to the same class of discrete choice models:

\begin{theorem}\label{thm_equivalence_RUM_and_stochastic_preferences}
A system of choice probabilities $\mathcal{P}$ is a RUM if and only if $\mathcal{P}$ is induced
by some stochastic preference model.
\end{theorem}
For a proof of Theorem~\ref{thm_equivalence_RUM_and_stochastic_preferences} see for instance~\citet{block1960random} or~\citet{koning2003discrete}.

It is a simple exercise to show that the discrete choice model of the tight example proposed in the proof of Theorem \ref{tightness_theorem} is in fact induced by a stochastic preference model. This implies that the three bounds presented in Section~\ref{section_general_performance_guarantees} are tight as well when restricted to RUMs. In particular, they are tight for discrete choice models with submodular demand functions, since RUMs have submodular demand functions.

\subsection{Distance based models}
An important class of discrete choice models that are induced by stochastic preferences are the {\DEF distance based models} \citep{murphy2003mixtures}. A distance based model is defined by a central ranking or preference $R$, a scale parameter $\theta \in \mathbb{R}_{+}$, and a distance function over the rankings $d: S_N \times S_N \to \mathbb{R}_+$. Then, the probability that the individual follows a ranking $r$ is given by
\begin{align*}
f(r|R,\theta) \coloneqq C(\theta) \exp[-\theta  d(r,R)],
\end{align*}
where $C(\theta)$ is a scaling constant that is chosen so that $f(r|R,\theta)$ is a probability distribution.

The most popular class of discrete choice models of this family are the Mallows models \citep{mallows1957non}, which are characterized as those distance based models in which the distance function $d(.,.)$ is the Kendall distance \citep{kendall1938new} (this distance function counts the pairwise disagreements between the rankings). Mallows models have been studied profoundly in voting contexts in the machine learning and statistics literature (see, e.g.~\citet{young1988condorcet} and \citet{diaconis1988group}). Very recently, \citet{jagabathula2015model} have used these models to understand and predict customer behaviour in the context of revenue management.

A natural extension of distance based models are the models obtained by the aggregation of multiple distance based models into a single stochastic preference. These models are known as {\DEF Mixture of Distance Based Models} (see, e.g. \citet{murphy2003mixtures} and \citet{awasthi2014learning}).

To the best of our knowledge, the assortment problem has not been studied in the literature under any distance based model. Since these models are stochastic preference models, they can be induced by a random utility model, and hence our revenue guarantees for the revenue-ordered assortment strategy hold under these models.

\subsection{Mixed Multinomial Logit}

\label{section_MMNL}
One of the most studied discrete choice models is the {\DEF multinomial logit} (MNL) model, first introduced by \citet{luce1959}.
The MNL model is a random utility based discrete choice model in which each product $x$ (including the no-purchase option) has utility $$U_x = v_x + \epsilon_x$$ where $v_x$ is a constant and all $\epsilon_x$ with $x \in \mathcal{C}$ are i.i.d.\ random variables with a Gumbel distribution with zero mean.

Without loss of generality, we may assume that the no-purchase option 0 has a mean utility of zero (i.e.\ $v_0=0$).
Under the MNL model, when an individual is shown a subset of products $S \subseteq \mathcal{C}$, the probability that she will choose product $x \in S \cup \{0\}$ is
$$\mathcal{P}(x,S) = \frac{e^{v_x}}{1 + \sum_{y \in S}e^{v_y}}.$$

The {\DEF Mixed Multinomial Logit} model is an extension of the MNL model in which the vector $V=(v_0,v_1,\dots,v_N)$ is no longer fixed but is now a random vector in $\mathbb{R}^{N+1}$ following some fixed distribution.  In the Mixed MNL model, when given a subset of products $S$, the probability that an individual chooses product $x \in S \cup \{0\}$ is then
$$\mathcal{P}(x,S) = \mathbb{E}\left[\frac{e^{v_x}}{1 + \sum_{y \in S}e^{v_y}}\right],$$
where the expectation is taken w.r.t.\ the random vector $V=(v_0,v_1,\dots,v_N)$.

Although the revenue-ordered assortments strategy is optimal under the MNL model \citep{talluri2004revenue},
this is no longer the case for the Mixed MNL model.
In fact, \cite{rusmevichientong2014assortment} gave an example
where the vector $V=(v_0,v_1,\dots,v_N)$ can only take two distinct values and yet the strategy is not optimal.
On the other hand, since every Mixed MNL model is also a random utility model, the guarantee from Theorem~\ref{theorem_revenue_ratio} applies:

\begin{corollary}
The revenue-ordered assortments strategy approximates the optimum revenue of the assortment problem under a Mixed MNL model to within a factor of $\frac{1}{(1 + \ln \rho)}$, where $\rho \coloneqq r_{k} / r_{1}$.
\end{corollary}

This improves the recent analysis by \citet{rusmevichientong2014assortment} of the revenue-ordered assortments strategy under a Mixed MNL model, who obtained an approximation factor of $\frac{1}{e(1 + \ln \rho)}$.

\subsection{Beyond random utility models} \label{section_beyond_rum}
In this section we briefly describe some discrete choice models considered in the literature that are not random utility models but still satisfy the regularity axiom.

\citet{fudenberg2015stochastic} proposed a choice model called {\em Additive Perturbed Utility} (APU) model, in which consumers are endowed with an utility function $u: \mathcal{C} \cup \{0\} \to \mathbb{R}$ over the alternatives (including the no-choice option) and a perturbation function that can reward choice randomisation. Specifically, this perturbation function $c:[0,1] \to \mathbb{R} \cup \{\infty\}$ is assumed to be strictly convex over $(0,1)$, and such that $\lim_{q \to 0}c'(q)=-\infty$.

For $s \in \mathbb{N}$ let $F(s) \coloneqq \{(p_0,p_1,\hdots,p_s) \in \mathbb{R}^{s+1} | p_i \geq 0$  for each $i \in \{0,1,\hdots,s\}$ and $\sum_{i=0}^s p_i=1 \}$.
Given a choice set $S \subseteq \mathcal{C}$ whose elements are enumerated as $a_1,a_2,\hdots,a_s$ in order w.r.t.\ $\mathcal{C}=\{1,2, \dots, N\}$
(i.e.\ $a_{i} < a_{i+1}$ for $i < s$), let $p^*(S) \in \mathbb{R}^{s+1}$ denote the point $p=(p_0,p_1,\hdots,p_s) \in F(s)$ maximising
$$\sum_{i=0}^s (u(a_{i})\cdot p_i-c(p_i)).$$
(As expected, $a_{0}=0$ denotes the no-choice option; also, we remark that $p^*(S)$ is uniquely defined, as follows from the strict convexity of the perturbation function $c$.)
The system of choice probabilities $\mathcal{P}$ of the model is then induced by these
vectors $p^*(S)$, by letting
$$(\mathcal{P}(0,S),\mathcal{P}(a_1,S), \hdots, \mathcal{P}(a_s,S)) \coloneqq p^*(S).$$

If the perturbation function is $c(x)= \alpha \cdot x \ln(x)$ with $\alpha > 0$ a fixed constant, then the model is equivalent to the Multinomial Logit model \citep{anderson1992discrete}. Although every APU model satisfies the regularity axiom (\citet[Theorem~1]{fudenberg2015stochastic}), it can be shown that there are APU models that are not RUM even when there are four alternatives in the universe $\mathcal{C}$ (\citet[Example~4]{fudenberg2015stochastic}).

Inspired by the experimental evidence that consumers do not pay attention to all alternatives in the choice set, \citet{aguiar2015stochastic} has recently introduced a new choice model called the {\em Hitting Fuzzy Attention Model} (H-FAM). Under H-FAM, the attention of the consumer to an alternative is not binary but can lie in a continuum between being not aware at all of the alternative, to being fully aware of it. To formally describe H-FAM we need to define a function called {\em substitutable attention capacity}. A substitutable attention capacity is a function $\phi: 2^{\mathcal{C}} \to [0,1]$ that is monotone, i.e.\ $\phi(A) \leq \phi(B)$ for all $A \subseteq B \subseteq \mathcal{C}$ and submodular, i.e.\ $\phi(A \cup\{x\}) - \phi(A) \geq \phi(B\cup \{x\}) - \phi(B)$ for all $A \subseteq B \subseteq \mathcal{C}$, and $x\in \mathcal{C}$. Intuitively, given a choice set $S\subseteq \mathcal{C}$, $\phi(S)$ represents the probability that the consumer would consider {\em at least} one alternative from $S$. The H-FAM is composed of a pair  $(\prec,\phi)$ where $\prec \in S_{N}$ is a strict preference order of the alternatives and $\phi$ is a substitutable attention capacity. We are now ready to define the system of choice probabilities for H-FAMs.

Suppose $(\prec,\phi)$ is defined as above. Then the system of choice probabilities $\mathcal{P}$ induced by this H-FAM is obtained by setting for $x \in \mathcal{C}$,
\begin{displaymath}
\mathcal{P}(x,S) \coloneqq \left\{ \begin{array}{ll}  \displaystyle \phi(\{x\} \cup \{y \in \mathcal{S} | y \prec x\}) - \phi(\{y \in \mathcal{S} | y \prec x\})  & \quad \textrm{if $x \in S$} \\[3ex]
0 & \quad \textrm{otherwise}\\
\end{array} \right.
\end{displaymath}
and letting $\mathcal{P}(0,S) \coloneqq 1 - \sum_{x\in S} \mathcal{P}(x,S)$.

The H-FAM, which contains as a special case the recent RUM choice model based on the bounded rationality proposed by \citet{manzini2014stochastic}, satisfies the regularity axiom. Nevertheless, \citet{aguiar2015stochastic} proved that H-FAM are not contained in RUM nor in APU.

Recently, \citet{mcclellon2015non} proposed another way to represent choice models based on a function $f: 2^{S_{N+1}} \to [0,1]$ such that $f(\emptyset)=0$, $f(S_{N+1})=1$ and $f(E) \leq f(F)$ when $E \subseteq F$. Thus, the domain of function $f$, known as the capacity function, is the collection of all subsets of strict preferences among the elements in $\mathcal{C} \cup \{0\}$. Given a strict preference $\prec \in S_{N+1}$ and a choice set $S \subseteq \mathcal{C}$, we say that $x \in S \cup \{0\}$ is $\prec${\em-optimal} w.r.t.\ $S$ if $x$ is preferred under $\prec$ among all other alternatives in $S \cup \{0\}$.

The choice model characterized by the capacity function $f$ is obtained by setting, for $x\in \mathcal{C} \cup \{0\}$,
\begin{displaymath}
\mathcal{P}(x,S) \coloneqq \left\{ \begin{array}{ll}  \displaystyle f(\{\prec \in S_{N+1} | x \textrm{ is } \prec\textrm{-optimal} \textrm{ w.r.t.\ } S \})  & \quad \textrm{if $x \in S \cup \{0\}$} \\[3ex]
0 & \quad \textrm{otherwise.}\\
\end{array} \right.
\end{displaymath}

While there exist choice models that are characterized by a capacity function $f$ but are not RUM, it can be shown that they all satisfy the regularity axiom. In fact, one of the main results of \citet{mcclellon2015non} is that a discrete choice model is regular if and only if it can be characterized by a capacity function $f$.

To the best of our knowledge, the assortment problem has not been studied in the literature under any of these models. Since all these models satisfy the regularity axiom, our revenue guarantees for the revenue-ordered assortment strategy hold.

\subsection{Envy-Free Pricing} \label{section_envy_free_pricing}

In this section we observe that certain envy-free pricing problems studied in theoretical computer science can be seen as special cases of the assortment problem described in Section~\ref{section_general_assortment_problem}.
The revenue-ordered assortments strategy then corresponds to the so-called uniform pricing strategy studied in that area.

In an envy-free pricing problem, it is assumed that there is a single seller (a {\em monopolist}) who has an unlimited supply of $n$ different types of products (or items) that are all offered to a set of $m$ consumers.
The seller assigns prices to the product types and then each consumer buys at most one product.
The seller's problem consists in choosing the prices so that the revenue obtained from the resulting sales is maximised. Naturally, this revenue depends on the behaviour of the consumers.
The corresponding pricing problems are called {\DEF unit demand envy-free pricing} ({\DEF UDP})  and differ only by their assumptions on the consumers'  behaviours.
The two main ones studied in the literature give rise to the  $UDP_{min}$ and $UDP_{rank}$ problems, which we describe shortly.
Before doing so, let us make a comment on the meaning of the adjective `envy-free'  appearing in these problems' names:
When the problems were first defined, the seller not only had to assign item prices, but also had to assign items to consumers, under the constraint that no consumer would have preferred receiving an item that was assigned to someone else, that is, the allocation should be {\em envy-free}.
Naturally, one can equivalently assume that each customer picks their preferred choice once prices are set, which is how the problems are usually phrased in the literature.

In the $UDP_{min}$ problem, each consumer $i \in [m]$ has an associated set $B_i \subseteq [n]$ of items that she is interested in buying, and a non-negative number $v_i$ (a {\DEF valuation}) which is the maximum price she is willing to pay to buy an item from that set.
Given a price assignment $p: [n] \rightarrow \mathbb{R}_{>0}$, consumer $i$ buys
the cheapest item from $B_i$ (breaking ties arbitrarily) if there is one with price at most $v_i$, or
none at all if there is none.

In the $UDP_{rank}$ problem, each consumer $i \in [m]$ has  an associated ranking $\phi_i: [n] \to [n]$ of the products and a non-negative number $v(i,x)$ (a {\DEF valuation})  for each product $x \in [n]$ modelling her willingness to pay for product $x$. Given a price assignment $p: [n] \rightarrow \mathbb{R}_{>0}$,  consumer $i$ then considers items in order of her preference list $\phi_i$ and then buys the first item $x$ that has price at most $v(i,x)$.

These two problems were introduced by Rusmevichientong in his Ph.D. thesis \citep{rusmevichientong2003non} and subsequently in \citep{rusmevichientong2006nonparametric}.
\citet{aggarwal2004algorithms} and \citet{guruswami2005profit} analysed the revenue guarantees of a simple pricing strategy called
{\DEF uniform pricing}: The seller put the same price $q$ on all products and chooses $q$ so
as to maximise the revenue. This can be seen as the revenue-ordered assortments strategy on an
auxiliary assortment problem, as will be explained in the next section.

We close this section by mentioning a variant of $UDP_{min}$ and  $UDP_{rank}$ that was also studied in the literature:
The seller is moreover required to choose a price assignment $p: [n] \rightarrow \mathbb{R}_{>0}$ that satisfies a  price ordering (or {\em price ladder}) of the $n$ items. This ordering, which is part of the problem instance, is given as a permutation $\psi: [n] \to [n]$. The price assignment $p$ chosen by the seller must then satisfy
$$p(\psi(x)) \leq p(\psi(y))$$ for every products $x,y \in [n]$ such that $\psi(x) < \psi(y)$.
The resulting problems are known as $UDP_{min}$ {\DEF with price ladder} ($UDP_{min}PL$) and $UDP_{rank}$ {\DEF with price ladder} ($UDP_{rank}PL$).
This variant was introduced by \citet{aggarwal2004algorithms}.

\subsubsection{$UDP_{min}$ as an assortment problem}
In this section we describe how the $UDP_{min}$ problem can be seen as an assortment problem under
some discrete choice model.
The main interest of this observation is that the resulting discrete choice model satisfies the regularity condition, axiom~\eqref{axiom:regularity}, and thus falls within the scope of the models
studied in this paper.

As before, suppose that there are $n$ products and $m$ consumers, and let $v_i$ denote the valuation of consumer $i \in [m]$.
For simplicity, we assume without loss of generality that consumers' valuations are such that $v_1 \leq v_2 \leq \cdots \leq v_m$.
First we start with a standard observation about the $UDP_{min}$ problem, namely, that in an optimal solution prices can be assumed to belong to the set of consumer valuations (a proof is given in the Appendix):

\begin{lemma} \label{lemma_pricing_to_values}
There exists an optimal price assignment $p: [n] \rightarrow \mathbb{R}_{>0}$
such that $p(x) \in \{v_1\dots,v_m\}$ for all $ x\in [n]$.
\end{lemma}

The following theorem shows that the $UDP_{min}$ problem is a special case of the assortment problem
under a regular discrete choice model.

\begin{theorem} \label{main_theorem_UDP_min}
Consider an instance of the $UDP_{min}$ problem with
$n$ products and $m$ consumers, and let $v_i$ denote the valuation of consumer $i \in [m]$.
Then one can define an instance of the assortment problem under a regular discrete choice model
(i.e.\ a finite set $\mathcal{C}$,
a revenue function $r: \mathcal{C} \rightarrow \mathbb{R}_{>0}$, and
a system of choice probabilities $\mathcal{P}$ satisfying axioms (\ref{discrete_ineq_1}), (\ref{discrete_eq_1}), (\ref{discrete_eq_2}) and (\ref{axiom:regularity}))
with the same optimal revenue. Moreover,
the uniform pricing and revenue-ordered assortments strategies on respective instances are
equivalent in the following sense:
\begin{itemize}
\item for each $i\in [m]$ there exists $S\subseteq \mathcal{C}$ consisting of all
products $y' \in \mathcal{C}$ with $r(y') \geq r(y)$ for some $y \in \mathcal{C}$ such that
the price assignment assigning price $v_i$ to each product $x\in [n]$
for the $UDP_{min}$ instance has the same revenue as the assortment $S$, and \\
\item for each $y\in \mathcal{C}$, there exists $i \in [m]$ such that
the assortment consisting of all products $y'\in \mathcal{C}$ with $r(y') \geq r(y)$ has the same
revenue as the price assignment assigning price $v_i$ to each product of the $UDP_{min}$ instance.
\end{itemize}
\end{theorem}

\begin{proof}
As before, we may assume that $v_1 \leq v_2 \leq \cdots \leq v_m$.
For each $i\in [m]$, let $B_i \subseteq [n]$ denote the subset of items consumer $i$ is interested in buying.
Define the set $\mathcal{C}$ of products for the assortment problem as follows:
$$\mathcal{C} \coloneqq [n] \times \{v_1,\dots,v_m\}.$$
Thus  $\mathcal{C}$ consists of all pairs of an item and a customer valuation.
The revenue function $r: \mathcal{C} \rightarrow \mathbb{R}_{>0}$
for the assortment problem is defined by setting
$$r( (x,v) ) \coloneqq m \cdot v$$
for all $(x, v) \in \mathcal{C}$.
(The purpose of the scaling factor $m$ is to cancel out the $1/m$ factor in the upcoming definition $\mathcal{P}$.)

Next we define the system of choice probabilities $\mathcal{P}$.
To do so, we first need to define the following sets:
$$
Q_i(S) \coloneqq \left\{
(x, v) \in S: x\in B_i, v \leq v_i, \textrm{ and } v' \geq v \; \; \forall
(x', v') \in S \textrm{ s.t.\ }x' \in B_i
\right\}
$$
for each $i\in [m]$ and $S\subseteq \mathcal{C}$.
Equipped with this notation, we define $\mathcal{P}$ as follows.
For each  $S\subseteq \mathcal{C}$ and $y\in \mathcal{C} \cup \{0\}$, let
\begin{displaymath}
\mathcal{P}(y,S) \coloneqq \frac{1}{m} \sum_{i=1}^m \mathcal{P}_i(y,S)
\end{displaymath}
where $\mathcal{P}_i$ ($i\in [m]$) is defined by setting:
\begin{displaymath}
\mathcal{P}_i((x,v),S) \coloneqq \left\{ \begin{array}{ll}  \displaystyle \frac{1}{|Q_i(S)|} & \quad \textrm{if $(x, v)\in Q_i(S)$} \\[3ex]
0 & \quad \textrm{otherwise}\\
\end{array} \right.
\end{displaymath}
for each $(x, v) \in \mathcal{C}$,
and
\begin{displaymath}
\mathcal{P}_i(0,S) \coloneqq 1 - \sum_{(x,v) \in S} \mathcal{P}_i((x,v),S).
\end{displaymath}

Let us prove that $\mathcal{P}$ is a regular discrete choice model.
Clearly  $\mathcal{P}(y, S) \geq 0$ for every $y\in \mathcal{C} \cup \{0\} $ and $S\subseteq \mathcal{C}$,
thus axiom~\eqref{discrete_ineq_1} is  satisfied.
If  $(x, v)\in \mathcal{C}$ and $S\subseteq \mathcal{C} \setminus \{(x,v)\}$  then
$\mathcal{P}((x, v), S) = 0$ since $(x, v) \notin Q_i(S)$ for each $i\in [m]$.
Hence,  axiom~\eqref{discrete_eq_1} is satisfied.
Also, for each $S\subseteq \mathcal{C}$ we have
$\sum_{(x,v) \in S}\mathcal{P}((x,v), S)  \leq 1$ since
$$\sum_{(x,v) \in S}\mathcal{P}_i((x,v), S)  \leq 1$$
for each $i \in [m]$. (In fact, the left-hand side
is  equal to either $0$ or $1$, depending on whether $Q_i(S)$ is empty or not.)
This implies that axiom~\eqref{discrete_eq_2} holds.
Therefore, it only remains to check axiom~\eqref{axiom:regularity}, the regularity condition.
Clearly, by the definition of $\mathcal{P}$, it is enough to show
that $\mathcal{P}_i$ satisfies axiom~\eqref{axiom:regularity} for each $i\in [m]$.
Let thus $i\in [m]$, let  $ S\subseteq S'\subseteq \mathcal{C} $, and let $y\in S \cup \{0\}$.
We wish to show that $\mathcal{P}_i(y, S) \geq \mathcal{P}_i(y, S')$.

First suppose that $y=(x, v) \in S$.
If $(x, v) \notin Q_i(S')$ then $\mathcal{P}_i((x,v), S')=0$ and
$\mathcal{P}_i((x,v), S) \geq \mathcal{P}_i((x,v), S')$ holds trivially,
so assume $(x, v) \in Q_i(S')$.
By the definition of $Q_i(S)$, it follows that $(x, v) \in Q_i(S)$ as well.
In fact, $Q_i(S) \subseteq Q_i(S')$ in this case.
This implies that
$\mathcal{P}_i((x,v), S) \geq \mathcal{P}_i((x,v), S')$, as desired.

Next, assume that $y=0$ (the no-choice option).
We will use the following observation:
\begin{displaymath} \label{obs_1_case_y_0}
\sum_{(x,v) \in T}\mathcal{P}_i((x, v), T) = \left\{ \begin{array}{ll}  \displaystyle 1 & \quad \textrm{if $Q_i(T) \neq \emptyset $} \\[1ex]
0 & \quad \textrm{otherwise}\\
\end{array} \right.
\end{displaymath}
for every $T \subseteq \mathcal{C}$.
If $Q_i(S) = \emptyset$, using this observation with $T=S$ we obtain
$$\mathcal{P}_i(0, S)
= 1 - \sum_{(x,v)\in S}\mathcal{P}_i((x, v), S)
= 1
\geq 1 - \sum_{(x,v)\in S'}\mathcal{P}_i((x, v), S')
= \mathcal{P}_i(0, S')
$$
as desired.
Now suppose that $Q_i(S) \neq \emptyset$.
Observe that this implies $Q_i(S') \neq \emptyset$ as well.
We then have
$$\mathcal{P}_i(0, S)
= 1 - \sum_{(x,v)\in S}\mathcal{P}_i((x, v), S)
= 1 - \sum_{(x,v)\in S'}\mathcal{P}_i((x, v), S')
= \mathcal{P}_i(0, S') $$
where the second equality follows from the above observation (with $T=S$ and $T=S'$).
Therefore, axiom~\eqref{axiom:regularity} is satisfied.

Next we prove that the maximum revenue achievable on each instance is the same.
Given a set $S\subseteq \mathcal{C}$ we define a corresponding price assignment
$p_S$ by setting
$$
p_S(x) \coloneqq \left\{ \begin{array}{ll}  \min\{v: (x, v) \in S\}  & \quad \textrm{if $\exists v$ s.t.\ $(x, v) \in S$} \\[1ex]
+ \infty & \quad \textrm{otherwise}\\
\end{array} \right.
$$
for each $x\in [n]$.
(Remark: If one wishes to insist on $p_S$ being real-valued, simply replace
 $+\infty$ in the above definition by any real larger than $v_m$.)
The revenue resulting from choosing assortment $S$ can be expressed as follows:
\begin{dmath*}[compact]
\sum_{y \in S}\mathcal{P}(y,S) \cdot r(y)
=  \frac{1}{m}\sum_{i \in [m]} \sum_{y \in S}\mathcal{P}_i(y,S) \cdot r(y)
= \sum_{i \in [m]} \sum_{(x,v) \in S}\mathcal{P}_i((x,v),S) \cdot v
=  \sum_{i \in [m]} \sum_{(x,v) \in Q_{i}(S)} \displaystyle \frac{1}{|Q_i(S)|}  \cdot v
=  \sum_{i \in [m], Q_{i}(S) \neq \emptyset} \min\{v: (x, v) \in S \textrm{ for some } x\in B_{i}\}
\end{dmath*}
Now, observe that $Q_{i}(S)$ is not empty if and only if there exists
a product $x\in B_{i}$ with price $p_{S}(x) \leq v_{i}$, that is,
if and only if customer $i$ buys some product in the $UDP_{min}$ instance with price assignment $p_{S}$.
Furthermore, if she does, then she buys some product $x'$ among the cheapest ones in $B_{i}$,
giving a revenue of $p_{S}(x') = \min\{v: (x', v) \in S\} = \min\{v: (x, v) \in S, x\in B_{i}\}$.
Thus, we deduce that the revenue of assortment $S$ is equal to the revenue resulting from price assignment $p_{S}$.

It follows that the maximum revenue achievable on this assortment problem instance  is at most that of the $UDP_{min}$ instance.
Furthermore, if $S \subseteq \mathcal{C}$
is an assortment consisting of all products $y'\in \mathcal{C}$ with $r(y') \geq r(y)$ for some $y\in \mathcal{C}$
then the corresponding price assignment $p_S$ satisfies
$p_S(x) = \min\{v: (x, v) \in S\} = p_S(x')$ for all $x, x'\in [n]$ and is thus uniform, as desired.
This shows one direction of the theorem.

To prove the other direction, suppose that
$p$ is a price assignment such that $p(x) \in \{v_1, \dots, v_m\}$ for all $x \in [n]$.
(Recall that there is an optimal price assignment of this form, by Lemma~\ref{lemma_pricing_to_values}.)
We define a corresponding assortment $S_p \subseteq  \mathcal{C}$ by setting
$$S_p \coloneqq \{ (x, v) : x \in [n], v\in\{v_1, \dots, v_m\}, v \geq p(x) \}.$$
Rewriting the revenue provided by assortment $S_p$ similarly as before, we see:
\begin{dmath*}[compact]
\sum_{y \in S}\mathcal{P}(y,S_p) \cdot r(y)
=  \frac{1}{m}\sum_{i \in [m]} \sum_{y \in S}\mathcal{P}_i(y,S_p) \cdot r(y)
=  \sum_{i \in [m]} \sum_{(x,v) \in S}\mathcal{P}_i((x,v),S) \cdot v
=  \sum_{i \in [m]} \sum_{x \in [n]}\mathcal{P}_i((x,p(x)),S_p) \cdot p(x)
=  \sum_{i \in [m], Q_{i}(S_p) \neq \emptyset}
{\min\{p(x): (x, p(x)) \in S \textrm{ for some } x\in B_{i}\} }
=  \sum_{i \in [m], Q_{i}(S_p) \neq \emptyset} \min\{p(x):  x\in B_{i}\}
\end{dmath*}

Observe that the last expression is exactly the revenue given by price assignment $p$.
Thus the optimal revenue for the $UDP_{min}$ instance is at most that of the assortment problem instance.
Since  by the previous paragraph it is also at least that, the two quantities are equal.
Moreover, if $p$ assigns the same price $v_i$ to all products $x \in [n]$ then
$S_p$ consists of all $y' \in \mathcal{C}$ with $r(y') \geq m \cdot v_i$, and is therefore of the desired form.
This concludes the proof.
\end{proof}

We note that \citet{aouad2015approximability} independently constructed a similar reduction to the one proposed above. Their objective is different from ours: they use the reduction to prove an APX-hardness result for the assortment problem under the RUM model even when there are only two different prices, whereas our objective is to observe the strong connection that exists between uniform pricing for the $UDP_{min}$ problem and revenue-ordered assortments in assortment optimisation.

Combining Theorem~\ref{main_theorem_UDP_min} with Theorem~\ref{theorem_revenue_ratio}, we obtain as a corollary a new revenue guarantee for the uniform pricing strategy.

\begin{corollary}
The uniform pricing strategy for the $UDP_{min}$ problem approximates the optimum revenue to within a factor of $\frac{1}{(1 + \ln \rho)}$, where $\rho \coloneqq  v_{m} / v_{1}$.
\end{corollary}

Theorem~\ref{main_theorem_UDP_min} together with Theorem~\ref{theorem_demand_ratio_in_optimal} yield the following bound, which was originally proved by \citet{aggarwal2004algorithms}.
The proof is given in the Appendix.

\begin{corollary}[\citet{aggarwal2004algorithms}]\label{static_pricing_known_guarantee_UDP_min_thm}
The uniform pricing strategy for the $UDP_{min}$ problem approximates the optimum revenue to within a factor of $\frac{1}{(1 + \ln m)}$.
\end{corollary}

\subsubsection{$UDP_{rank}$ as an assortment problem}
In this section we note that the $UDP_{rank}$ problem can also be seen as an assortment problem under a regular discrete choice model.
As in the previous section, suppose that there are $n$ products and $m$ consumers. In an optimal solution of the $UDP_{rank}$ problem, prices can always be assumed to belong to the set of consumer valuations. The proof is a straightforward adaptation of that of Lemma~\ref{lemma_pricing_to_values} and is thus omitted.

\begin{lemma} \label{lemma_pricing_to_values_udp_rank}
  There exists an optimal price assignment $p: [n] \rightarrow \mathbb{R}_{>0}$
  such that $p(x)\in  \{v(i,x) : i \in [m] , x \in [n] \} $ for all $ x\in [n]$.
\end{lemma}

The following theorem shows that the $UDP_{rank}$ problem is a special case of the assortment problem under a regular discrete choice model. It can be proved along the same lines as the proof of Theorem~\ref{main_theorem_UDP_min}. To keep the paper concise, we leave the proof to the reader.

\begin{theorem} \label{main_theorem_UDP_rank}
Consider an instance of the $UDP_{rank}$ problem with $n$ products and $m$ consumers, and let $v(i,x)$ denote the valuation of consumer $i \in [m]$ for product $x \in [n]$. Then one can define an instance of the assortment problem under a regular discrete choice model
(i.e.\ a finite set $\mathcal{C}$, a revenue function $r: \mathcal{C} \rightarrow \mathbb{R}_{>0}$, and
a system of choice probabilities $\mathcal{P}$ satisfying axioms (\ref{discrete_ineq_1}), (\ref{discrete_eq_1}), (\ref{discrete_eq_2}) and (\ref{axiom:regularity})) with the same optimal revenue. Moreover,
the uniform pricing and revenue-ordered assortments strategies on respective instances are equivalent in the following sense:
\begin{itemize}
\item for each $(i,x)\in [m]\times[n]$ there exists $S\subseteq \mathcal{C}$ consisting of all
products $y' \in \mathcal{C}$ with $r(y') \geq r(y)$ for some $y \in \mathcal{C}$ such that
the price assignment of setting price $v(i,x)$ to every item of the $UDP_{rank}$ instance has the same revenue as the assortment $S$, and \\
\item for each $y\in \mathcal{C}$, there exists $(i,x) \in [m]\times[n]$ such that the assortment consisting of all products $y'\in \mathcal{C}$ with $r(y') \geq r(y)$ has the same revenue as the price assignment that sets price $v(i,x)$ to each item of the $UDP_{rank}$ instance.
\end{itemize}
\end{theorem}

Thanks to Theorem~\ref{main_theorem_UDP_rank}, we know that using the uniform pricing strategy on an instance $I$ of the $UDP_{rank}$ problem has the same performance as the revenue-ordered assortment strategy on the instance $I'$ of the assortment problem associated to $I$. Combining this with Theorem~\ref{theorem_revenue_ratio}, we obtain the following corollary:

\begin{corollary}\label{theorem:udp_rank_approximation}
The uniform pricing strategy approximates the $UDP_{rank}$ problem to within a factor of $\frac{1}{(1 + \ln \rho)}$, where $\rho \coloneqq \max\{ v(i,e) : i \in [m] , e \in [n] \} / \min \{v(i,e) : i \in [m] , e \in [n] \}$.
\end{corollary}

Similarly as for the $UDP_{min}$ problem, we can also apply Theorem~\ref{theorem_demand_ratio_in_optimal} to derive the following bound that is a function of the number of consumers, which was already established by~\citet{aggarwal2004algorithms}:

\begin{corollary}[\citet{aggarwal2004algorithms}]
The uniform pricing strategy approximates the $UDP_{rank}$ problem to within a factor of $\frac{1}{(1 + \ln m)}$.
\end{corollary}

We end this section with the remark that our results on $UDP_{min}$ and $UDP_{rank}$ also hold for the variant of these two problems involving a price ladder (as defined at the beginning of this section).
This is because adding the price ladder constraint can only decrease the seller's optimum revenue and the uniform pricing strategy always satisfies the price ladder constraint.

\subsection{Stackelberg pricing}
In this section we consider a pricing problem called Stackelberg Minimum Spanning Tree problem that was introduced by  \citet{cardinal2011stackelberg}. We show that this problem can be restated as an assortment problem under a specific discrete choice model satisfying the regularity condition. Furthermore, the so-called uniform pricing algorithm for the Stackelberg Minimum Spanning Tree problem corresponds then to revenue-ordered assortments. This connection allows us to see the results on uniform pricing obtained by \citet{cardinal2011stackelberg} as being a special case of the approximation guarantees of revenue-ordered assortments established in Section~\ref{section_general_performance_guarantees}.
(We note that the results from~\cite{cardinal2011stackelberg} are currently the best known approximation factors for the Stackelberg Minimum Spanning Tree problem.)

An instance of the {\DEF Stackelberg Minimum Spanning Tree problem} consists of an (undirected, simple) graph $G$,
a bipartition of the edges of $G$ into a set $R$ of {\DEF red edges} and a set $B$ of {\DEF blue edges}, and
a cost function $c:R \to \mathbb{R}_{>0}$ assigning a positive cost to each red edge.
The objective is to choose a price assignment $p:B \to \mathbb{R}_{>0}$ for the blue edges
so that the revenue resulting from a consumer buying a minimum weight spanning tree of $G$ is maximised.
The latter revenue is the sum of the prices $p(e)$ of all blue edges $e\in B$ that appear in the spanning tree.

Let us remark that if there is no spanning tree of $G$ consisting only of red edges, then the optimal revenue is unbounded.
Indeed, the customer is then forced to buy at least one blue edge, and thus one could price all blue edges arbitrarily high, knowing that at least one will be bought by the customer.
To avoid such trivialities, we always assume that there exists a red spanning tree in the instance under consideration.

As is well known, a minimum weight spanning tree can be obtained by the greedy (a.k.a.\ Kruskal's) algorithm, which consists in first ordering the edges in non-decreasing order of costs and then considering each edge in order, and selecting it if it does not create a cycle with edges that have already been selected.
Moreover, every minimum weight spanning tree can be obtained this way, by selecting an adequate ordering of the edges (the freedom in choosing the ordering being how ties are broken for edges having the same cost).
Thus we may assume that the customer builds her minimum weight spanning tree by running the greedy algorithm following some ordering of the edges.
While we do not necessarily know the latter ordering completely (in case there are ties), it will be assumed that the customer always gives priorities to blue edges over red edges in case of ties.
This technical assumption is needed to make sure that the revenue resulting from pricing the blue edges is independent of the particular spanning tree bought by the customer.
An informal justification for this assumption is that in case of a tie between a blue edge and red edge, we could decrease the price of the blue edge by an arbitrarily small amount to make sure it is considered first.

The purpose of this section is to prove the following theorem.

\begin{theorem} \label{main_theorem_StackMST}
Consider an instance of the {\StackMST} problem, consisting of graph $G$,
a bipartition of the edges into a set $R$ of red edges and a set $B$ of blue edges, and
a cost function $c:R \to \mathbb{R}_{>0}$.
Let $c_{1}, \dots, c_{k}$ denote the different values taken by the cost function,
in non-decreasing order.
Then one can define an instance of the assortment problem under a regular discrete choice model
(i.e.\ a finite set $\mathcal{C}$,
a revenue function $r: \mathcal{C} \rightarrow \mathbb{R}_{>0}$, and
a system of choice probabilities $\mathcal{P}$ satisfying axioms (\ref{discrete_ineq_1}), (\ref{discrete_eq_1}), (\ref{discrete_eq_2}) and (\ref{axiom:regularity}))
with the same optimal revenue. Moreover,
the uniform pricing and revenue-ordered assortments strategies on respective instances are
equivalent in the following sense:
\begin{itemize}
\item for each $i\in [k]$ there exists $S\subseteq \mathcal{C}$ consisting of all
products $y' \in \mathcal{C}$ with $r(y') \geq r(y)$ for some $y \in \mathcal{C}$ such that
the price assignment assigning price $c_i$ to each edge $e\in B$
for the  {\StackMST}  instance has the same revenue as the assortment $S$, and \\
\item for each $y\in \mathcal{C}$, there exists $i \in [k]$ such that
the assortment consisting of all products $y'\in \mathcal{C}$ with $r(y') \geq r(y)$ has the same
revenue as the price assignment assigning price $c_i$ to each blue edge of the  {\StackMST}  instance.
\end{itemize}
\end{theorem}

Theorem~\ref{main_theorem_StackMST} can in fact be proved in the more general
setting of matroids and the corresponding {\em \StackMatroid} problem, as we now explain.
First we recall basic definitions regarding matroids.
A matroid is a pair $(E, \mathcal{X})$ with $E$ a finite set of elements and $\mathcal{X}$ a collection of subsets of $E$ called {\DEF independent sets} that satisfy the following three properties:
\begin{itemize}
\item $\emptyset \in \mathcal{X}$;
\item if $X \in \mathcal{X}$ and $Y \subseteq X$ then $Y \in \mathcal{X}$, and
\item if $X, Y \in \mathcal{X}$ with $|X| < |Y|$ then there exists
$y\in Y-X$ such that $X \cup \{y\} \in \mathcal{X}$.
\end{itemize}

An inclusion-wise maximal independent set of a matroid is said to be a {\DEF base} of the matroid.
Note that all bases have the same cardinality, as follows from the above axioms.

Given a linear ordering $L$ of the elements of a matroid $M=(E, \mathcal{X})$, the {\DEF greedy algorithm} computes a base of $M$ using $L$ as follows: Enumerating the elements of $E$ as $e_1, e_2, \dots, e_m$ according to $L$,
the greedy algorithm defines $m+1$ independent sets $I_0, I_1, \dots, I_m$ inductively, by first setting $I_0 \coloneqq \emptyset$, and then
for each $i=1, \dots, m$, by setting $I_i \coloneqq I_{i-1} \cup \{e_i\}$ if $I_{i-1} \cup \{e_i\}$ is independent (i.e.\ if it is in $\mathcal{X}$), and $I_i \coloneqq I_{i-1}$ otherwise.
The algorithm then outputs $I_m$, the last independent set it computed.
It is easily seen that the latter is a base of the matroid.

It is sometimes convenient to run the greedy algorithm on a subset $F$ of the elements of the matroid $M$ under consideration. The behaviour is exactly the same as described above but with respect to the elements $e_1, \dots, e_k$
of $F$ ordered according to the ordering induced by $L$.
The resulting independent set is of course not necessarily a base of $M$, though it is of maximal size among independent sets contained in $F$.

The following lemma is a well-known and fundamental property of the greedy algorithm on matroids.
We provide a proof in the Appendix, to keep the paper self-contained.
(For more background results on matroids and the greedy algorithm, the reader is referred to the textbook of~\cite{S03B}.)

\begin{lemma}
\label{lem:greedy_matroid}
Let $M=(E, \mathcal{X})$ be a matroid and let $L$ be a linear ordering of the elements of $E$.
For $F \subseteq E$, let $\greedy_{M}(F, L)$ denote the independent subset of $F$ obtained by running the greedy algorithm on the set $F$ using ordering $L$.
Then for every $F \subseteq F' \subseteq E$ the following two properties hold:
\begin{enumerate}[(i)]
\item $|\greedy_{M}(F', L)| \geq |\greedy_{M}(F, L)|$, and \\[.1ex]
\item $F \cap \greedy_{M}(F', L) \subseteq \greedy_{M}(F, L)$.
\end{enumerate}
\end{lemma}

The {\StackMST} problem is a special case of the {\StackMatroid} problem, where
we are given a matroid $M=(E, \mathcal{X})$,
a bipartition of the elements set $E$ into a set $R$ of red elements and a set $B$ of blue elements, and
a cost function $c:R \to \mathbb{R}_{>0}$.
The objective is to choose a price assignment $p:B \to \mathbb{R}_{>0}$ for the blue elements
so that the revenue resulting from a consumer buying a minimum-weight basis of $M$ is maximised.
The latter revenue is the sum of the prices $p(e)$ of all blue elements $e\in B$ that appear in the basis.

Exactly as for {\StackMST} problem, in a {\StackMatroid} instance it is always assumed that
there exists a basis of $M$ that consists only of red elements, since otherwise the optimum
revenue is unbounded.
As recalled earlier, once prices are set for blue elements,
a minimum-weight basis of $M$ can be computed using the greedy algorithm with an
ordering of the elements in $R\cup B$ that is compatible with these costs/prices.
In case of ties we assume that priority is given to blue elements over red elements,
for the same reason as before.

We remark that, given a price assignment $p:B \to \mathbb{R}_{>0}$,
there might be more than one possible ordering of the elements in $R \cup B$
compatible with these prices, and we do not know which one will be used
by the customer when computing her minimum-weight basis.
Nevertheless, the resulting revenue is always the same
(and in particular the optimum revenue is well defined).
This follows from the following lemma, whose proof can be found in the Appendix.

\begin{lemma}
\label{lem:greedy_welldefined}
Consider an instance of the {\StackMatroid} problem, consisting of matroid $M=(E, \mathcal{X})$,
a bipartition of the elements set $E$ into a set $R$ of red elements and a set $B$ of blue elements, and a cost function $c:R \to \mathbb{R}_{>0}$.
Consider some price assignment $p:B \to \mathbb{R}_{>0}$.
Then for every two linear orderings $L$ and $L'$ of the elements in $R \cup B$ that are compatible
with these costs and prices, and for every $\gamma \in \R$, we have
\begin{dmath*}[compact]
|\{e\in B\cap \greedy_{M}(R \cup B, L):  p(e)=\gamma\}| =
|\{e\in B\cap \greedy_{M}(R \cup B, L'):  p(e)=\gamma\}|.
\end{dmath*}
In particular, the revenue resulting from price assignment $p$ is independent of the
particular ordering used by the customer.
\end{lemma}

As mentioned earlier,
it turns out that Theorem~\ref{main_theorem_StackMST} can be proved for any {\StackMatroid} problem,
that is, the only property that is really needed is that the set of forests of a graph form a matroid (its {\em graphical matroid}).
Moreover, we find it easier to use the language of matroids in the proof. 
We thus prove the following generalisation of Theorem~\ref{main_theorem_StackMST}, see the Appendix for the proof.

\begin{theorem} \label{main_theorem_StackMatroid}
Consider an instance of the {\StackMatroid} problem, consisting of matroid $M=(E, \mathcal{X})$,
a bipartition of the elements set $E$ into a set $R$ of red elements and a set $B$ of blue elements, and
a cost function $c:R \to \mathbb{R}_{>0}$.
Let $c_{1}, \dots, c_{k}$ denote the different values taken by the cost function,
in non-decreasing order.
Then one can define an instance of the assortment problem under a regular discrete choice model
(i.e.\ a finite set $\mathcal{C}$,
a revenue function $r: \mathcal{C} \rightarrow \mathbb{R}_{>0}$, and
a system of choice probabilities $\mathcal{P}$ satisfying axioms (\ref{discrete_ineq_1}), (\ref{discrete_eq_1}), (\ref{discrete_eq_2}) and (\ref{axiom:regularity}))
with the same optimal revenue. Moreover,
the uniform pricing and revenue-ordered assortments strategies on respective instances are
equivalent in the following sense:
\begin{itemize}
\item for each $i\in [k]$ there exists $S\subseteq \mathcal{C}$ consisting of all
products $y' \in \mathcal{C}$ with $r(y') \geq r(y)$ for some $y \in \mathcal{C}$ such that
the price assignment assigning price $c_i$ to each element $e\in B$
for the  {\StackMatroid}  instance has the same revenue as the assortment $S$, and \\
\item for each $y\in \mathcal{C}$, there exists $i \in [k]$ such that
the assortment consisting of all products $y'\in \mathcal{C}$ with $r(y') \geq r(y)$ has the same
revenue as the price assignment assigning price $c_i$ to each blue element of the  {\StackMatroid}  instance.
\end{itemize}
\end{theorem}

We note that specialising Theorems~\ref{theorem_revenue_k}, \ref{theorem_revenue_ratio}, and~\ref{theorem_demand_ratio_in_optimal} to the Stackelberg Minimum Spanning Tree problem via Theorem~\ref{main_theorem_StackMST}, we obtain exactly the three bounds on the uniform pricing algorithm proved by~\citet{cardinal2011stackelberg} (see Theorem~3 in that paper). As shown by the latter authors, these three bounds are already tight in this specific setting.

The reader familiar with polymatroids will undoubtedly have noticed that
the proofs of the results presented in this section extend directly to the setting of polymatroids.
This is because all properties of the greedy algorithm we used hold more generally in that setting.
Therefore, Theorem~\ref{main_theorem_StackMatroid} remains
true for the natural extension of the {\StackMatroid} problem to polymatroids.
We nevertheless decided to present the material of this section in the language of matroids
to simplify the exposition,
the reader interested in the definition and properties of polymatroids is referred to~\cite{S03B}.

\section{Nesting-by-fare-order property under the multi-period setting} \label{sec:nesting-by-fare}
In this short section we consider the multi-period model introduced by \citet{talluri2004revenue} and analyse the solution structure of the assortments offered as the available capacity and remaining time decreases. We prove that the two monotonicity results that were originally shown by \citet{rusmevichientong2014assortment} for Mixed MNL models naturally extend to the case of regular discrete choice models. As discussed in  \citet{rusmevichientong2014assortment}, these results could potentially be used in the implementation of standard revenue management systems.

First, let us define formally the multi-period model of \citet{talluri2004revenue}: It is assumed that sales occur over a finite horizon of $T$ time periods and that there is a maximum number $Q$ of items that the firm can sell along the complete time horizon. For a concrete example, suppose that the firm is selling seats on a flight leg ($Q$ represents the number of seats) and that items in $\mathcal{C}$ represent the different fare classes, each having an associated price (or revenue). At each time period, based on the time periods remaining and the available capacity, the firm selects a choice set $S \subseteq \mathcal{C}$ to offer to customers. Then, a single customer arrives and decides which item to buy from $S$ (if any).

Let us assume that consumers follow a regular discrete choice model, and let $\mathcal{P}$ denote the corresponding system of choice probabilities. In order to simplify notations, we further assume that items in $\mathcal{C}=\{1,\dots,N\}$ are sorted in non-increasing order of revenue, that is, $r(i) \geq r(i+1)$ for each $i\in \{1,\dots,N-1\}$.
We suppose from now on that the firm always offers revenue-ordered assortments,  during the whole time horizon.

In the discussion below, it will be more convenient to measure the remaining time before the end of the time horizon than the time elapsed since the beginning.
Suppose thus that we are $t \in [T]$ periods prior to the end of the time horizon.
Assume further that the available inventory of items that the firm can sell is $q \in [Q]$.
As in Section \ref{section_general_performance_guarantees}, let $r_1,r_2,\dots,r_k$ be the distinct values taken by the revenue function $r$, sorted in non-decreasing order (therefore, $r(1)=r_k$ and $r(N)=r_1$).
Let $j(\ell)$ be the index such that  $\{1, 2, \dots, j(\ell)\}$ is the set of the items with revenue at least $r_{\ell}$, for $\ell =1,\dots,k$.
Define $\mathcal{J}_t	(q)$ as the expected revenue of the revenue-ordered strategy over the remaining $t$ periods, given that there are $q$ units left in the inventory at the current time period.
Let also $\mathcal{J}_t(q, \ell)$ denote the expected revenue obtained, given that there are $q$ units left in the inventory at the current time period, and we offer the set consisting of all items with revenue at least $r_{\ell}$ at the current time period, and then proceed with the revenue-ordered strategy for the remaining time periods.
Thus $$\mathcal{J}_t(q) = \max_{\ell \in [k]} \mathcal{J}_t(q,\ell).$$
Moreover, if  $t > 0$ and $q > 0$, then $\mathcal{J}_t(q,\ell)$ satisfies the following relation:
\begin{dmath*}[compact]
 \mathcal{J}_t(q,\ell) = \sum_{x=1}^{j(\ell)} \mathcal{P}(x,\{1,\dots,j(\ell)\})\cdot (r(x) + \mathcal{J}_{t-1}(q-1)) +
 \mathcal{P}(0,\{1,\dots,j(\ell)\}) \cdot \mathcal{J}_{t-1}(q).
\end{dmath*}
(As expected, we set $\mathcal{J}_t(q, \ell) = 0$ in case $q=0$ or $t=0$.)

Let $\ell^*_t(q) := \min \{\ell \in [k]  :  \mathcal{J}_t(q,\ell) = \mathcal{J}_t(q)\}$.

The following theorem, which generalises Theorem 6 in \citet{rusmevichientong2014assortment}, establishes the monotinicity properties of revenue ordered assortments alluded to at the beginning of this section.
In words, it shows that the function $\ell^*$ is non-increasing in the remaining capacity and non-decreasing in the time remaining.

\begin{theorem} \label{nesting_by_fare-order_theorem}
Let  $t \in [T]$ and let $q \in [Q]$.
Then
\begin{itemize}
\item $\ell^*_t(q) \leq \ell^*_{t}(q-1)$ if $q \geq 2$, and \\[0.5ex]
\item $\ell^*_t(q) \geq \ell^*_{t-1}(q)$ if $t \geq 2$.
\end{itemize}
\end{theorem}

The proof of Theorem~\ref{nesting_by_fare-order_theorem} is given in the Appendix.

\section{Final remarks}
In this paper we studied the performance of a simple and well-known heuristic for the assortment problem, known as revenue-ordered assortment. We provided three worst-case performance guarantees that are written as a function of product prices as well as the distribution of product purchases in an optimal solution. An appealing feature of our performance guarantees is that they simply rely on the following condition known as regularity: the probability of selecting a specific product from the set being offered cannot increase if the set is enlarged. This condition is satisfied by almost all discrete choice models considered in the revenue management and choice theory literature, and in particular by the well-known random utility model (RUM). Furthermore, we show that the three bounds are exactly tight, even when one restricts to the RUM subfamily.

We found a remarkable connection between two pricing problems called {\em unit demand envy-free pricing} and {\em Stackelberg minimum spanning tree} studied in the theoretical computer science literature and the assortment problem under a regular discrete choice model. In essence, we show that both pricing problems can be seen as assortment problems under discrete choice models satisfying the regularity condition. In addition, the revenue-ordered assortments strategy correspond then to the well-studied {\em uniform pricing} heuristic. Interestingly, the bounds from the revenue-ordered assortment are able to match and unify known results on uniform pricing that were proved separately in the literature for the envy-free pricing problems and the Stackelberg minimum spanning tree problem.

\section*{Acknowledgments}
We thank Victor Aguiar, Adrian Vetta, and Gustavo Vulcano for their insightful comments that greatly improved the paper.
We are also grateful to the anonymous referees for their careful reading of the paper and helpful remarks.

\appendix

\bibliographystyle{plainnat}
\bibliography{references_discrete_choice}

\section*{Appendix A}

See Table~\ref{tab:list_of_models} for a (non-exhaustive) list of discrete choice models that satisfy the regularity axiom.

\begin{table}[h!]
  \centering
  \label{tab:list_of_models}
  \begin{tabular}{|l|l|}
    \hline
    \textbf{Discrete Choice Model} & \textbf{Reference} \\
    \hline
    MNL &  \citet{luce1959} \\
    \hline
    Mixed MNL &  \citet{rusmevichientong2014assortment} \\
    \hline
    Markov Chain Model &  \citet{blanchet2013markov} \\
    \hline
    Mallows Model &  \citet{mallows1957non} \\
    \hline
    Mixture of Distance Based Models & \citet{murphy2003mixtures} \\
    \hline
    Stochastic Preference & \citet{block1960random} \\
    \hline
    Random Utility Model & \citet{thurstone1927law} \\
    \hline
    Bounded Accumulation Model & \citet{webb2016dynamics} \\
    \hline
    Additive Perturbed Utility (*) & \citet{fudenberg2015stochastic} \\
    \hline
    Hitting Fuzzy Attention Model (*) & \citet{aguiar2015stochastic} \\
    \hline
    Non-Additive Random Utility (*) & \citet{mcclellon2015non} \\
        \hline
  \end{tabular}
  \caption{List of discrete choice models that satisfy the regularity axiom (non-exhaustive). Models marked with an asterisk are not RUM.}
\end{table}

\section*{Appendix B}

In this section we provide the proofs missing from the main text.

\begin{proof}[Proof of Lemma~\ref{lem:submodularity}]
Recall from Section \ref{sec:rum} that every discrete choice model based on random utility can also be described by means of a probability distribution over all rankings of the elements in $\mathcal{C} \cup \{0\}$ (see for instance~\citet{block1960random}). 
Let $\mathcal{S}_{N+1}$ denote the set of all such rankings, and let $\alpha_i$ be the probability of choosing the ranking $\prec_i \in \mathcal{S}_{N+1}$. Then the system of choice probabilities for our RUM can be written as

\begin{eqnarray*}
P(x,S)=\sum_{\prec_i \in \mathcal{S}_{N+1}}\alpha_i \mathds{1} \{ x \prec_i y \quad \forall y \in (S \setminus\{x\}) \cup\{0\}\}
\end{eqnarray*}
for all $S \subseteq \mathcal{C}$ and  $x \in S$. 

Consider now, for every ranking $\prec_i \in \mathcal{S}_{N+1}$ , the function $D_i(S):2^{\mathcal{C}}\to\{0,1\}$ defined as
\begin{eqnarray*}
D_i(S)= \sum_{x \in S} \mathds{1} \{ x \prec_i y \quad \forall y \in (S \setminus\{x\}) \cup\{0\}\}.
\end{eqnarray*}
Observe that $D_i(S)=0$ if $\prec_i$ puts element $0$ before all elements in $S$, and $D_i(S)=1$ otherwise. 
Let us show that $D_i(S)$ is submodular. That is, for every $S \subseteq S' \subseteq \mathcal{C}$ and $x \in \mathcal{C}$,
\begin{eqnarray}\label{submodularity_inequality}
D_i(S'\cup \{x\}) - D_i(S') \leq  D_i(S\cup \{x\}) - D_i(S).
\end{eqnarray}
First, observe that $D_i$ is monotone (i.e.\ $D_i(S) \leq D_i(S')$ if $S \subseteq S'$). Second, observe that if the LHS in \eqref{submodularity_inequality} is equal to zero, the inequality directly follows. Consider then the case in which the LHS is equal to 1.
This implies that $D_i(S')=0$, and thus $D_i(S)=0$, and that $D_i(S'\cup\{x\})=1$. Hence, $x \prec_i y$ for all $y \in S' \cup \{0\}$, which means that $x \prec_i y$ for all $y \in S \cup \{0\}$. Therefore, $D_i(S\cup \{x\})=1$, and the inequality \eqref{submodularity_inequality} follows.

To finish the proof of the lemma, observe that
\begin{eqnarray*}
f(S) &=& \sum_{x \in S} \mathcal{P}(x,S) \\
&=& \sum_{x \in S} \sum_{\prec_i \in \mathcal{S}_{N+1}}\alpha_i \mathds{1} \{ x \prec_i y \quad \forall y \in (S \setminus\{x\}) \cup\{0\}\} \\
&=& \sum_{\prec_i \in \mathcal{S}_{N+1}}\alpha_i \sum_{x \in S} \mathds{1} \{ x \prec_i y \quad \forall y \in (S \setminus\{x\}) \cup\{0\}\} \\
&=& \sum_{\prec_i \in \mathcal{S}_{N+1}}\alpha_i D_i(S).
\end{eqnarray*}
Since the sum of submodular functions is also submodular, the proof is now complete.
\end{proof}

\begin{proof}[Proof of Lemma~\ref{lemma_pricing_to_values}]
Let $p: [n] \rightarrow \mathbb{R}_{>0}$ denote an optimal price assignment and suppose there is an item $x\in [n]$ such that $p(x)=q \notin \{v_1,\dots,v_m\}$.  If no consumer buys product $x$ at price $q$ then the seller would obtain the same revenue if the price for product $x$ were changed to $v_m$, as is easily checked. If, on the other hand, there is a consumer that buys product $x$ at price $q$, we do the following modification.
Let $w$ denote the lowest consumer valuation that is greater or equal to $q$, that is, $w = \min \{ v_j : j \in [m] \text{ and } v_j \geq q \}$. Suppose now that the seller modifies the price assignment $p$ by setting $p(x) \coloneqq w$.
First, observe that every consumer buying a product different from $x$ is unaffected by this modification since the price of product $x$ has not decreased. Second, every consumer that was buying product $x$ before
either still buys it under the new price assignment, or buys another product $y$ from her set with price
$q \leq p(y) \leq w$.
Hence, the seller's revenue does not decrease.
Repeating this argument at most $n$ times, we eventually obtain  a price assignment that satisfies this lemma.
\end{proof}

\begin{proof}[Proof of Corollary~\ref{static_pricing_known_guarantee_UDP_min_thm}]
To be precise, this corollary follows from the proof of
Theorem~\ref{main_theorem_UDP_min}, as we now explain.
Consider the assortment problem instance associated to a given $UDP_{min}$ instance  described in that proof.
As in Theorem \ref{theorem_demand_ratio_in_optimal}, let $S^{*} \subseteq \mathcal{C}$ denote an optimal solution to this instance and let $$N_i\coloneqq \sum_{\substack{x\in S^{*}, \\ r(x) \geq r_{i}}} \mathcal{P}(x, S^{*})$$
for each $i \in [k]$.
For convenience, let us scall these quantities by a factor $m$:
For each $i \in [k]$, let $$\hat{N}_i\coloneqq N_i \cdot m.$$
Then $\hat{N}_1$ is exactly the number of consumers of the $UDP_{min}$ instance that buy some item in the optimal solution. Thus,
\begin{eqnarray}\label{N1_less_eq_than_m}
\hat{N}_1 \leq m.
\end{eqnarray}
Moreover, letting $\ell \in [k]$ be maximum such that $\hat{N}_{\ell} > 0$, we have that $\hat{N}_{\ell}$ is the number of consumers that buy the most expensive item that was sold in the optimal solution. Hence,
\begin{eqnarray}\label{ell_at_least_one}
\hat{N}_{\ell} \geq 1.
\end{eqnarray}
By Theorem \ref{theorem_demand_ratio_in_optimal}, the revenue-ordered assortments strategy approximates the optimum to within a factor of
$$
\frac{1}{1 + \ln (N_{1} / N_{\ell})}
= \frac{1}{1 + \ln (\hat{N}_{1} / \hat{N}_{\ell})}
\leq \frac{1}{1 + \ln m},
$$
where the last inequality follows from \eqref{N1_less_eq_than_m} and \eqref{ell_at_least_one}).
The claim then follows from Theorem~\ref{main_theorem_UDP_min}.
\end{proof}

\begin{proof}[Proof of Lemma~\ref{lem:greedy_matroid}]
Let $G \coloneqq \greedy_{M}(F, L)$ and $G'\coloneqq \greedy_{M}(F', L)$.
To see the first property, suppose for a contradiction that
$|G'| < |G|$.
Enumerate the elements of $F'$ as $e_1, \dots, e_k$ following the ordering $L$.
Let $e_i \in G - G'$ be such that $G' \cup \{e_i\}$ is independent.
In particular, $(G' \cap \{e_1, \dots, e_{i-1}\}) \cup \{e_i\}$ is also independent.
This implies that when greedy considered the $i$-th element of $F'$, $e_i$, it added $e_i$ in the set being built, a contradiction.

Let us now prove the second property.
Enumerate the elements of $F$ this time as $e_1, \dots, e_k$ according to $L$.
Arguing by contradiction, suppose that $e_i \in G' - G$ for some $i$.
We may assume that index $i$ is minimum with this property.
Thus, $G' \cap \{e_1, \dots, e_{i-1}\} \subseteq G$.

Let $I\coloneqq G \cap \{e_1, \dots, e_{i-1}\}$.
Thus $I \cup \{e_i\}$ is not independent, by definition of the greedy algorithm.
Let $K$ be an independent set satisfying $I \subseteq K \subseteq I \cup G'$ and inclusion-wise maximal with this property.
Let us point out that $e_i \notin K$, since $I \subseteq K$.

{\bf Case~1: $|K| < |G'|$.}
Since $G'$ is independent there exists $e\in G'-K$ such that $K\cup \{e\}$ is independent, which contradicts the maximality of $K$.

{\bf Case~2: $|K| > |G'|$.}
Since $K$ is independent there exists $e\in K- G' = I - G'$ such that $G'\cup \{e\}$ is independent.
However, when building $G'$ the greedy algorithm did not pick element $e$ because the subset $J \subseteq G'$ of elements it already picked when considering $e$ was such that $J \cup \{e\}$ is not independent.
Since $J \cup \{e\} \subseteq G'\cup \{e\}$, this contradicts the fact that
$G'\cup \{e\}$ is independent.

{\bf Case~3: $|K| = |G'|$.}
Here, we use that $|K| > |G' - \{e_i\}|$, which implies that there exists
$e\in K - (G' - \{e_i\})$ such that  $(G' - \{e_i\}) \cup\{e
\}$ is independent.
Since $e_i \notin K$, we have $K - (G' - \{e_i\}) = K - G' = I - G'$, and
it follows that $e=e_j$ for some $j <i$, meaning that $e_j$ was considered before $e_i$ by the greedy algorithm when building set $G'$.
Yet the algorithm did not pick $e_j$ when it considered $e_j$, because the subset $J \subseteq G'$ of elements it already picked at that time could not be extended into an independent set by adding $e_j$.
However, $J \cup \{e\} \subseteq (G' - \{e_i\}) \cup\{e_{j}
\}$, contradicting the independence of $(G' - \{e_i\}) \cup\{e_{j}
\}$.

Since in each of the three cases we derived a contradiction, we deduced that the second property holds, as claimed.
\end{proof}

\begin{proof}[Proof of Lemma~\ref{lem:greedy_welldefined}]
This is a consequence of the following more general
property of the greedy algorithm:
Suppose that $E$ is partitioned into $\ell$ blocks
$E_1, \dots, E_{\ell}$, and that $L$ and $L'$ are two linear orderings of
the elements in $E$ that agree on this block partition, in the sense that
$e <_{L} f$ and $e <_{L'} f$ for all $i\in \{1, \dots, k-1\}$ and
$e\in E_i, f\in E_{i+1}$.
Then
\begin{dmath}[compact]
\label{eq:greedyprop}
|\greedy_{M}(E, L) \cap E_i| =
|\greedy_{M}(E, L') \cap E_i|
\quad \forall i  \in \{1, \dots, k\}.
\end{dmath}
To see that it implies the lemma, it suffices to take the partition of $E$ obtained by
partitioning $R$ according to the costs of the red elements
and $B$ according to the prices of the blue elements, and ordering the resulting blocks in
non-decreasing order of costs/prices, giving priority to blue over red in case of ties.

Thus it remains to prove~\eqref{eq:greedyprop}.
Arguing by contradiction, suppose that the property does not hold
and let $i$ be the smallest index such that
$|\greedy_{M}(E, L) \cap E_i| \neq |\greedy_{M}(E, L') \cap E_i|$.
Say without loss of generality
$|\greedy_{M}(E, L) \cap E_i| < |\greedy_{M}(E, L') \cap E_i|$.
By our choice of $i$,
$|\greedy_{M}(E, L) \cap (E_1 \cup \cdots \cup E_{i-1})| =
|\greedy_{M}(E, L') \cap (E_1 \cup \cdots \cup E_{i-1})|$.
Since $|\greedy_{M}(E, L) \cap (E_1 \cup \cdots \cup E_{i})| <
|\greedy_{M}(E, L') \cap (E_1 \cup \cdots \cup E_{i})|$,
by the last of the matroid axioms there exists
$e \in (\greedy_{M}(E, L') - \greedy_{M}(E, L)) \cap (E_1 \cup \cdots \cup E_{i}) $
such that $(\greedy_{M}(E, L) \cap (E_1 \cup \cdots \cup E_{i})) \cup \{e\}$
is independent.
Hence, the greedy algorithm w.r.t.\ ordering $L$
could have picked $e$ when considering the elements
in $E_i$ but did not, a contradiction.
\end{proof}

\begin{proof}[Proof of Theorem~\ref{main_theorem_StackMatroid}]
Define the set $\mathcal{C}$ of products for the assortment problem as follows:
$$\mathcal{C} \coloneqq B \times \{c_1,\dots,c_k\}.$$

Thus  $\mathcal{C}$ consists of all pairs of a blue element and a cost $c_i$ taken by some red element.
The revenue function $r: \mathcal{C} \rightarrow \mathbb{R}_{>0}$
for the assortment problem is defined by setting
$$r( (e, q) ) \coloneqq |B| \cdot q$$
for all $(e, q) \in \mathcal{C}$.

Next we define the system of choice probabilities $\mathcal{P}$.
To do so, it is convenient to first define an auxiliary matroid $M'$ with ground set
$R \cup \mathcal{C}$, where $X \subseteq R \cup \mathcal{C}$ is independent if only if
the following two conditions are satisfied:
\begin{enumerate}[(1)]
\item $|\{(e, q): q\in \{c_{1}, \dots, c_{k}\} \}| \leq 1$ for each element $e \in B$, and
\item $(R\cap X) \cup \{e \in B: (e, q) \in \mathcal{X} \textrm{ for some } q\in \{c_{1}, \dots, c_{k}\}\}$
is independent in $M$.
\end{enumerate}
We leave it to the reader to check that $M'$ is indeed a matroid.
In order to give some intuition about $M'$, we remark that in the case of the {\StackMST} problem,
where $M$ is the graphical matroid of the input graph, $M'$ is simply the graphical matroid of the graph
obtained by replacing each edge $e$ with $k$ parallel copies of $e$.

For simplicity, let us say that the {\DEF cost} of element $x \in R \cup \mathcal{C}$
is $c(x)$ if $x\in R$, and $q$ if $x=(e, q) \in \mathcal{C}$.
Let $L$ be a linear ordering of the elements in $R \cup \mathcal{C}$
that is consistent with their associated costs (elements appear in non-decreasing order of cost),
where priority is given to elements in $\mathcal{C}$ over those in $R$
(that is, if $(e, q) \in \mathcal{C}$, $f\in R$, and $c(f) = q$, then $(e, q) <_{L} f$).

We define $\mathcal{P}$ as follows.
For each  $S\subseteq \mathcal{C}$ and $(e,q)\in \mathcal{C}$, let

\begin{displaymath}
\mathcal{P}((e,q),S) \coloneqq \left\{ \begin{array}{ll}  \displaystyle \frac{1}{|B|} & \quad \textrm{if $(e, q)\in
\greedy_{M'}(R\cup S, L)$} \\[3ex]
0 & \quad \textrm{otherwise}\\
\end{array} \right.
\end{displaymath}
and
\begin{displaymath}
\mathcal{P}(0,S) \coloneqq 1 - \sum_{(e,q) \in S} \mathcal{P}((e,q),S).
\end{displaymath}

Let us prove that $\mathcal{P}$ is a regular discrete choice model.
Clearly  $\mathcal{P}(y, S) \geq 0$ for every $y\in \mathcal{C} \cup \{0\} $ and $S\subseteq \mathcal{C}$,
thus axiom~\eqref{discrete_ineq_1} is  satisfied.
If  $(e, q)\in \mathcal{C}$ and $S\subseteq \mathcal{C} \setminus \{(e,q)\}$  then
$\mathcal{P}((e, q), S) = 0$ since $(e, q) \notin \greedy_{M'}(S ,L)$ trivially.
Hence,  axiom~\eqref{discrete_eq_1} is satisfied.
Also, for each $S\subseteq \mathcal{C}$ we have
$\sum_{(e,q) \in S}\mathcal{P}((e,q), S)  \leq |\greedy_{M'}(R\cup S, L)| / |B| \leq  1$,
since $|\greedy_{M'}(R\cup S, L)| \leq |B|$.
This implies that axiom~\eqref{discrete_eq_2} holds.
Therefore, it only remains to check axiom~\eqref{axiom:regularity}, the regularity condition.
Let thus $S\subseteq S'\subseteq \mathcal{C}$, and let $y\in S \cup \{0\}$.
We want to show that $\mathcal{P}(y, S) \geq \mathcal{P}(y, S')$.

First suppose that $y=(e, q) \in S$.
If $\mathcal{P}((e,q), S') = 0$, then trivially $\mathcal{P}((e,q), S) \geq \mathcal{P}((e,q), S')$.
If $\mathcal{P}((e,q), S') = 1/|B|$, then $(e,q) \in \greedy_{M'}(R\cup S', L)$.
By Lemma~\ref{lem:greedy_matroid}, we also have $(e,q) \in \greedy_{M'}(R\cup S, L)$,
and thus $\mathcal{P}((e,q), S) = 1/|B|$.
Hence, $\mathcal{P}((e,q), S) \geq \mathcal{P}((e,q), S')$ holds in both cases.

Now assume that $y=0$ is the no-choice option.
We know from Lemma~\ref{lem:greedy_matroid} that
$(R \cup S) \cap \greedy_{M'}(R \cup S', L) \subseteq \greedy_{M'}(R \cup S, L)$.
Since $|\greedy_{M'}(R \cup S', L)| \geq |\greedy_{M'}(R \cup S, L)|$ by the same lemma,
and since $(R \cup S') - (R \cup S) = S' - S \subseteq  \mathcal{C}$,
we deduce that
$$|\greedy_{M'}(R \cup S', L) \cap \mathcal{C}| \geq |\greedy_{M'}(R \cup S, L) \cap \mathcal{C}|.$$
This implies that
$\sum_{(e,q)\in S'}\mathcal{P}((e, q), S') \geq \sum_{(e,q)\in S}\mathcal{P}((e, q), S)$,
and we conclude that
$$\mathcal{P}(0, S)
= 1 - \sum_{(e,q)\in S}\mathcal{P}((e, q), S)
\geq 1 - \sum_{(e,q)\in S'}\mathcal{P}((e, q), S')
= \mathcal{P}(0, S').
$$
Therefore, axiom~\eqref{axiom:regularity} is satisfied.

Next we prove that the maximum revenue achievable on each instance is the same.
Given a set $S\subseteq \mathcal{C}$ we define a corresponding price assignment
$p_S$ by setting
$$
p_S(e) \coloneqq \left\{ \begin{array}{ll}  \min\{q: (e, q) \in S\}  & \quad \textrm{if $\exists q$ s.t.\ $(e, q) \in S$} \\[1ex]
+ \infty & \quad \textrm{otherwise}\\
\end{array} \right.
$$
for each element $e \in B$.
The revenue resulting from choosing assortment $S$ can be expressed as follows:
\begin{dmath*}[compact]
\sum_{y \in S}\mathcal{P}(y,S) \cdot r(y)
=  |B| \sum_{(e,q) \in S}\mathcal{P}((e,q),S) \cdot q
=  \sum_{(e, q)\in \greedy_{M'}(R\cup S, L)}  q.
\end{dmath*}

In the {\StackMatroid} instance with price assignment $p_S$,
the customer buys element $e\in B$ if and only if the greedy algorithm
chooses element $e$ when considering matroid $M$ with some linear ordering
$L^*$ of its elements $R\cup B$ that is consistent with their costs/prices, where
priority is given to elements in $B$ in case of ties.
That is, letting $c'(e) \coloneqq c(e)$ if $e\in R$ and
$c'(e) \coloneqq p_S(e)$ if $e\in B$, the ordering
$L^*$ used by the customer satisfies:
\begin{dgroup}
\begin{dmath}
\label{eq:ordering1}
\text{if $c'(e) < c_S(f)$ then $e <_{L^*} f$ for every $e,f \in R \cup B$, and}
\end{dmath}
\begin{dmath}
\label{eq:ordering2}
\text{if $c'(e) = c_S(f)$ then $e <_{L^*} f$ for every $e \in B$ and $f\in R$.}
\end{dmath}
\end{dgroup}
There might be more than one linear ordering satisfying the above properties;
however thanks to Lemma~\ref{lem:greedy_welldefined} we know that all such orderings yield the
same revenue.
Since we are not interested in the particular elements bought by the customer but only
by the resulting revenue, for the purpose of the following analysis we may assume that $L^{*}$ `breaks ties' in the same that $L$ does, that is,
\begin{dgroup*}
\begin{dmath*}
{\forall e,f\in B: e  <_{L^*} f \Leftrightarrow (e, p_S(e)) <_L (f, p_S(f));}
\end{dmath*}
\begin{dmath*}
{\forall e,f\in R: e  <_{L^*} f \Leftrightarrow e <_L f.}
\end{dmath*}
\end{dgroup*}
Element $e\in B$ is bought by the customer
if and only if $e\in  \greedy_{M}(R\cup B, L^*)$, which by our assumption on $L^{*}$
is the same condition as $(e, p_S(e))\in \greedy_{M'}(R\cup S, L)$.
Furthermore, if $(e, q)\in \greedy_{M'}(R\cup S, L)$ for some $q$ then this
$q$ is unique and by definition of the price functin $p_S$ we have $p_S(e)=q$.
Hence, the revenue resulting from price function $p_S$ equals
\begin{dmath*}
\sum_{e\in  \greedy_{M}(R\cup B, L^*)} p_S(e)
= \sum_{(e, q)\in \greedy_{M'}(R\cup S, L)}  q,
\end{dmath*}
and is thus equal to the revenue given by assortment $S$.
This shows that the maximum revenue achievable on the {\StackMatroid} instance
is at least that achievable on the assortment problem we defined.
We now prove the converse statement, and hence that the two quantities are the same.

Let $p$ be some price function for the elements in $B$ and let
$L^*$ be a linear ordering of the elements in $R \cup B$
that a customer could use when running the greedy algorithm on $M$
under this price assignment.
As recalled above, $L^*$ is thus any linear ordering satisfying
\eqref{eq:ordering1} and \eqref{eq:ordering2},
where in this case $c'(\cdot)$ is modified by setting
$c'(e) \coloneqq p(e)$ if $e\in B$.

We may assume that all prices assigned by $p$ to elements in $B$
are in the set $\{c_1, \dots, c_k\} \cup \{+\infty\}$.
Indeed, if $c_{i-1} < p(e) < c_i$ for some $e\in B$ and $i\in \{1, \dots, k\}$
(where we let $c_0:=0$) then we can increase $p(e)$ to $c_i$
without changing $L^*$ being a valid ordering for these prices,
which can only improve the resulting revenue.
Similarly, if $p(e) > c_k$ then we may as well set $p(e) \coloneqq +\infty$
since element $e$ will never be bought by the customer,
as follows from the existence of a base of $M$ in $R$.
Since our goal is to bound from above the revenue resulting from price assignment $p$
by the optimal revenue achievable on the assortment problem, we may thus
iteratively  modify $p$ as described until it has the desired form.

Relying on the fact that $p$ takes values in  $\{c_1, \dots, c_k\} \cup \{+\infty\}$,
we define the following corresponding assortment $S \subseteq \mathcal{C}$:
\[
S \coloneqq \{(e, p(e)): e\in  \greedy_{M}(R\cup B, L^*) \}.
\]
The ordering $L^*$ of the elements in $R \cup B$
induces in a natural way an ordering $L$ of the elements in $R\cup S$.
The revenue given by assortment $S$ is then
\begin{dmath*}[compact]
\sum_{y \in S}\mathcal{P}(y,S) \cdot r(y)
=  |B| \sum_{(e,q) \in S}\mathcal{P}((e,q),S) \cdot q
=\sum_{(e,q)\in  \greedy_{M'}(R\cup S, L)} q
=\sum_{e\in  \greedy_{M}(R\cup B, L^*)} p(e)
\end{dmath*}
and is thus equal to the revenue resulting from price function $p$.
This shows that the optimal revenue achievable on the assortment problem is least
that achievable on the  {\StackMatroid} instance, as desired.
\end{proof}

Next we provide a proof of  Theorem~\ref{nesting_by_fare-order_theorem}.
In order to do so, we need to introduce some definitions and a lemma.

Given $\delta \in \R$ such that $r_k + \delta \geq 0$, we define $\mathcal{L}^*(\delta)$ to be the set of indices $\ell \in [k]$ with the property that, if we add $\delta$ to the revenue of all items, then choosing all items of revenue at least $r_{\ell} + \delta$ gives an optimal revenue-ordered assortment (w.r.t.\ the modified revenue function) for the usual assortment problem. (Note that whenever the revenue of an item becomes negative, the item will never be chosen by the revenue ordered assortment strategy.)
That is, $\mathcal{L}^*(\delta)$ is the set of indices $\ell \in [k]$  such that
$$ \sum_{x=1}^{j(\ell)}\mathcal{P}(x,\{1,\dots,j(\ell)\})(r(x) + \delta) = \max_{\ell' \in [k]}\sum_{x=1}^{j(\ell')}\mathcal{P}(x,\{1,\dots,j(\ell')\})(r(x) + \delta).$$

\begin{lemma} \label{lemma_multi_period}
Let $\delta_1, \delta_2 \in \mathbb{R}$ with $\delta_1 + r_k \geq 0$ and $\delta_1 \leq \delta_2$. Then, $\min \mathcal{L}^*(\delta_2) \leq \min \mathcal{L}^*(\delta_1)$.
\end{lemma}
\begin{proof}

Let $\ell_1:= \min\mathcal{L}^*(\delta_1)$, $\ell_2 := \min \mathcal{L}^*(\delta_2)$ and $\Delta := \delta_2 - \delta_1 \geq 0$.
For the purpose of contradiction, suppose that $\ell_2 > \ell_1 $. This means that $ j(\ell_2) < j(\ell_1)$ since products in $\mathcal{C}$ are enumerated in decreasing order of revenue.

We will show that if we add $\delta_2$ to the revenue of all items (w.r.t.\ revenue function $r$), the expected revenue resulting from the assortment $\{1,\dots,j(\ell_2)\}$ is at most that of the assortment $\{1,\dots, j(\ell_1)\}$. This implies $\ell_2 \leq \ell_1$, a contradiction.

We have
\begin{align*}
& \sum_{x=1}^{j(\ell_2)} \mathcal{P}(x, \{1,\dots,j(\ell_2)\}) (r(x)+ \delta_2) \\
 & \quad \quad \quad = \sum_{x=1}^{j(\ell_2)} \mathcal{P}(x, \{1,\dots,j(\ell_2)\}) (r(x)  + \delta_1) + \Delta \sum_{x=1}^{j(\ell_2)} \mathcal{P}(x, \{1,\dots,j(\ell_2)\}) \\
& \quad \quad \quad \leq   \sum_{x=1}^{j(\ell_1)} \mathcal{P}(x, \{1,\dots,j(\ell_1)\}) (r(x) + \delta_1) + \Delta \sum_{x=1}^{j(\ell_2)} \mathcal{P}(x, \{1,\dots,j(\ell_2)\}) \\
 & \quad \quad \quad \leq   \sum_{x=1}^{j(\ell_1)} \mathcal{P}(x, \{1,\dots,j(\ell_1)\}) (r(x)  + \delta_1) + \Delta \sum_{x=1}^{j(\ell_1)} \mathcal{P}(x, \{1,\dots,j(\ell_1)\}) \\
& \quad \quad \quad =   \sum_{x=1}^{j(\ell_1)} \mathcal{P}(x, \{1,\dots,j(\ell_1)\}) (r(x)+ \delta_2)
\end{align*}
The first inequality holds because $\{1,\dots,j(\ell_1)\}$, by definition, yields the highest expected revenue among all revenue-ordered assortments when $\delta_1$ is added to the revenue of each item.
The second inequality follows from Lemma~\ref{lemma_property_regularDCM} and the assumption that $\ell_2 > \ell_1$.
This concludes the proof.
\end{proof}

\begin{proof}[Proof of~Theorem \ref{nesting_by_fare-order_theorem}]
We begin by proving that $\ell^*_t(q) \leq \ell^*_{t}(q-1)$ if $q \geq 2$.
For $t' \geq 0$ and $q' \geq 1$ let $$\Delta \mathcal{J}_{t'}(q') := \mathcal{J}_{t'}(q') - \mathcal{J}_{t'}(q'-1),$$ that is, $\Delta \mathcal{J}_{t'}(q')$ is the marginal value of the capacity when there are $t'$ time periods remaining.
(Let us point out that $\Delta \mathcal{J}_{t'}(q') = \mathcal{J}_{t'}(q') = \mathcal{J}_{t'}(q'-1)=0$ if $t'=0$.)

We can express $\mathcal{J}_t(q)$ as follows:
\begin{eqnarray}
\mathcal{J}_t(q) &=& \max_{\ell \in [k]} \left\{ \sum_{x=1}^{j(\ell)} \mathcal{P}(x,\{1,\dots,j(\ell)\})(r(x) + \mathcal{J}_{t-1}(q-1)) + \mathcal{P}(0,\{1,\dots,j(\ell)\})  \mathcal{J}_{t-1}(q) \right\} \nonumber \\
& = & \max_{\ell \in [k]} \left\{ \sum_{x=1}^{j(\ell)} \mathcal{P}(x,\{1,\dots,j(\ell)\}) (r(x) - \Delta \mathcal{J}_{t-1}(q)) \} + \mathcal{J}_{t-1}(q) \right\} \label{nesting_by_fare-order_theorem_proof_label_1}
\end{eqnarray}

Observe that
\begin{equation} \label{eq_th_5.1_0}
\ell^*_t(q-1) = \min \mathcal{L}^*(- \Delta \mathcal{J}_{t-1}(q-1)).
\end{equation}

In other words, $\ell^*_t(q-1)$ is the largest revenue ordered assortment which is optimal when $\Delta \mathcal{J}_{t-1}(q-1)$ is subtracted from the revenue of each item. Since the most revenue one could obtain from an extra unit of capacity is the revenue of the most expensive product (i.e.\ $r_k$), we have that $r_k - \Delta \mathcal{J}_{t-1}(q-1) \geq 0$ as needed.

Suppose that we subtract $\Delta \mathcal{J}_{t-1}(q)$ to the original revenue of each item. By Equation (\ref{nesting_by_fare-order_theorem_proof_label_1}) we know that

\begin{equation}\label{eq_th_5.1_1}
\ell^*_t(q) = \min \mathcal{L}^*(-\Delta \mathcal{J}_{t-1}(q))
\end{equation}

Again, since the most revenue one could obtain from an extra unit of capacity is the revenue of the most expensive product (i.e. $r_k$), we have that $r_k - \Delta \mathcal{J}_{t-1}(q) \geq 0$ as desired.

\citet[Lemma~4]{talluri2004revenue} proved that $\Delta \mathcal{J}_{t-1}(q-1)  \geq \Delta \mathcal{J}_{t-1}(q)$ always holds, regardless of the discrete choice model under consideration (that is, the three axioms \eqref{discrete_ineq_1}, \eqref{discrete_eq_1} and \eqref{discrete_eq_2} are enough for this property to hold). Therefore, by Lemma \ref{lemma_multi_period}, we have that
\begin{equation}\label{eq_th_5.1_2}
\min \mathcal{L}^*(-\Delta \mathcal{J}_{t-1}(q-1)) \geq \min \mathcal{L}^*(-\Delta \mathcal{J}_{t-1}(q))
\end{equation}

Combining  (\ref{eq_th_5.1_1}), (\ref{eq_th_5.1_2}) and (\ref{eq_th_5.1_0}) we have that
\begin{eqnarray}
\ell^*_t(q) &= \min \mathcal{L}^*(-\Delta \mathcal{J}_{t-1}(q)) \leq  \min \mathcal{L}^*(-\Delta \mathcal{J}_{t-1}(q-1)) \nonumber = \ell^*_t(q-1)
\end{eqnarray}
as desired.

We now proceed to prove that $\ell^*_t(q)$ is non-increasing in $t$, which carries on in a similar way.

Suppose that we subtract $\Delta \mathcal{J}_{t-2}(q)$ to the original revenue of each item. By equation (\ref{nesting_by_fare-order_theorem_proof_label_1}) we know that
\begin{equation}\label{eq_th_5.1_4}
\ell^*_{t-1}(q) = \min \mathcal{L}^*(- \Delta \mathcal{J}_{t-2}(q))
\end{equation}

Again, since the most revenue one could obtain from an extra unit of capacity is the revenue of the most expensive product (i.e. $r_k$), we have that $r_k - \Delta \mathcal{J}_{t-2}(q) \geq 0$ as desired.

\citet[Lemma~5]{talluri2004revenue} proved that $\Delta \mathcal{J}_{t-1}(q)  \geq \Delta \mathcal{J}_{t-2}(q)$ always holds, regardless of the discrete choice model under consideration. Therefore, by Lemma \ref{lemma_multi_period}, we have that
\begin{equation}\label{eq_th_5.1_5}
\min \mathcal{L}^*(-\Delta \mathcal{J}_{t-2}(q)) \leq \min \mathcal{L}^*(-\Delta \mathcal{J}_{t-1}(q))
\end{equation}

Combining  (\ref{eq_th_5.1_4}), (\ref{eq_th_5.1_5}) and (\ref{eq_th_5.1_1}) we have that
\begin{eqnarray}
\ell^*_{t-1}(q) = \min \mathcal{L}^*(- \Delta \mathcal{J}_{t-2}(q)) \leq  \min \mathcal{L}^*(-\Delta \mathcal{J}_{t-1}(q))  \nonumber = \ell^*_t(q)
\end{eqnarray}
as desired.
\end{proof}

\end{document}